%% file: main_isit.tex
\documentclass[conference,letterpaper]{IEEEtran}

\usepackage{bm,bbm,amsmath,subcaption, amssymb,amsthm,mathtools,style,diagbox,multirow,multicol}
\usepackage[table]{xcolor}
\usepackage{libertine}
\usepackage{hyperref}
\usepackage[nameinlink,capitalise]{cleveref}
\usepackage{setspace}
\definecolor{linkc}{rgb}{0.1, 0.5, 0.7}
\definecolor{citec}{rgb}{0.6, 0.3, 0.7}
\definecolor{urlc}{rgb}{0.5, 0.1, 0.2}
\hypersetup{
    colorlinks=true,
    linkcolor=linkc,
    citecolor=citec,
    urlcolor=urlc
}

\usepackage[utf8]{inputenc} 
\usepackage[T1]{fontenc}
\usepackage{url}
\usepackage{ifthen}
\usepackage{cite}

\newtheorem*{lemma*}{Lemma}
\newtheorem*{theorem*}{Theorem}

\begin{document}
\title{Optimal moments on redundancies in job cloning} 


\author{%
  \IEEEauthorblockN{Sahasrajit Sarmasarkar}
  \IEEEauthorblockA{Department of Electrical Engineering \\
                    Stanford University\\
                    Email: sahasras@stanford.edu}
  \and
  \IEEEauthorblockN{Harish Pillai}
  \IEEEauthorblockA{Department of Electrical Engineering \\
  Indian Institute of Technology, Bombay\\ 
            Email: hp@ee.iitb.ac.in}
}

\maketitle


\begin{abstract}
  We consider the problem of job assignment where a master server aims to compute some tasks and is provided a few child servers to compute under a uniform straggling pattern where each server is equally likely to straggle. We distribute tasks to the servers so that the master is able to receive most of the tasks even if a significant number of child servers fail to communicate. We first show that all \textit{balanced} assignment schemes have the same expectation on the number of distinct tasks received and then study the variance. We show constructions using a generalization of ``Balanced Incomplete Block Design''\cite{doi:10.1111/j.1469-1809.1939.tb02219.x,sprott1955} minimizes the variance, and constructions based on repetition coding schemes attain the largest variance. 

\end{abstract}

\input{test-intro-isit}

\input{premiliminary-isit}
\input{mean-and-variance-isit}

\input{extreme-varaince-results-isit}

\input{random-return}
\input{conclusion}

\section*{Acknowledgment}

Most of this work was performed when SS was an undergraduate at Indian Institute of Technology, Bombay. SS would like to thank Siddharth Chandak for insightful discussions on Section \ref{sec_randomness_return_servers}.

\bibliographystyle{IEEEtran}
\bibliography{refs}  

\onecolumn  
\appendix
\setcounter{section}{0}  
\renewcommand{\thesection}{\Alph{section}}  

\input{mean_theorem_proof}
\input{variance_computation}
\input{lemma_proximal_minimal}
\input{min_var_proof}
\input{max_variance_proof}
\input{g0_expr_proof}
\input{g_p_diff_proof}
\input{max_min_var_sampled}

\input{appendix-corollary}

\end{document}

%% file: test-intro-isit.tex
\section{Introduction}

A distributed computing framework, widely used for large-scale tasks \cite{AWScitation}, utilizes multiple machines (workers) to perform computations. In a typical setup, the master server divides tasks among workers (also called child servers), which then compute and send results back to the master, which aggregates them to complete the overall task. In noisy communication environments, some servers (workers), known as stragglers, may be significantly slower than others. The master must rely on the faster, non-straggling workers, although these may not be known in advance. Lately, there has been work to mitigate the issue of stragglers by introducing redundancies \cite{JMLR:v20:18-148, DBLP:journals/corr/abs-1906-10664,10.1145/3055281}. 

Coding theoretic techniques have often been used to introduce redundancies for straggler mitigation described extensively in \cite{li2020coded} and split the data to assign different data parts to different servers. Distributed coding framework has been used in gradient computation \cite{pmlr-v70-tandon17a}, matrix-matrix multiplication \cite{lee2017speeding, yu2017polynomial}, polynomial computation \cite{yu2019lagrange} and convolution coding\cite{8006960} using techniques from coding theory. 


On a side note, there may be several scenarios where it may not be possible to split a computing task to multiple sub-parts, and job cloning is often used for straggler mitigation in such scenarios as studied in \cite{chen2014queueing,ananthanarayanan2013effective,joshi2015queues,10.1145/3055281,4154088}. This is also popular in cloud computing \cite{10.1145/3055281,10.1145/2847220.2847223,10.1145/2637364.2592042} where one task is assigned to multiple servers to combat the stragglers to obtain a low compute time. 

In this paper, we study the problem of assigning multiple cloned jobs to identical servers, where each server has a high likelihood of straggling. This may also serve as a reasonable model for an erasure communication channel \cite{Shannon1948} between the child servers and the master server.


\section{Our contributions}

Consider a scenario with $n$ jobs and 
$c$ identical (child) servers, where jobs are replicated to mitigate the impact of slow (straggling) workers. We focus on a homogeneous setup, where each server is assigned $k$ jobs and each job is assigned to $r$ servers, with all servers equally likely to straggle. We assume that non-straggling workers can successfully transmit their tasks to the master, as studied in \cite{pmlr-v70-tandon17a,10.1145/3055281,8849684}. To this end, we analyze the mean and variance of the number of distinct jobs received, which is detailed in Section \ref{setup_section}. Our contributions are outlined as follows.

a) When every server is equally and independently likely to straggle with the same probability, then for every balanced assignment, the expected number of completed jobs that the master receives is the same (Theorem \ref{all_mean_sampled}). 

b) We show that certain special balanced assignments (called proximally compact and stretched compact designs) when they exist, are guaranteed to attain the least and largest variance respectively (Theorem \ref{max_min_var_sampled}).

\subsection{Roadmap}

We define the notations in Section \ref{setting_section} and address the assignment problem in Section \ref{setup_section} under the assumption that any set of $x$ servers is equally likely to straggle, analyzing assignment policies for maximal and minimal variance. In Section \ref{sec_randomness_return_servers}, we extend the analysis to show that the extremal variance results hold even when 
$x$ is sampled from a distribution which accurately models our setup of straggling servers.


\section{Related work}

Block designs, widely used in experimental design \cite{bose1939partially,addelman1969generalized,shieh2004effectiveness}, group experiments into blocks and apply random treatments to each. Recently, variants of block designs have been used in constructing LDPC codes \cite{ammar2004construction}, gradient coding \cite{kadhe2019gradient,sakorikar2021variants}, and error-correcting codes \cite{smith1968application}. The most common are 2-designs, or balanced incomplete block designs (BIBDs), where every pair of points appears together in the same number of blocks \cite{bose1939partially,10.5555-1202540}. These 2-designs are known to uniquely achieve A-, D-, and E-optimality in experimental design \cite{10.1214/aos/1176343002,4e05890a-5af4-3aa4-948b-f44a2cc9151f}, a result further generalized in \cite{7c66992f-2b5b-3d69-be0f-fe6ad7e70150}. In contrast, our framework focuses on minimizing the variance in the number of distinct tasks assigned to the master, using repetition coding and a generalized form of BIBDs to achieve the largest and smallest variances, respectively.


Task assignment policies in distributed server systems have been widely studied, with \cite{semchedine2011task} providing a comprehensive survey. Typically, tasks arrive stochastically and are distributed to minimize response time, as explored in \cite{colajanni1998analysis,harchol2000task,harchol1999choosing}. Recent studies show that task distributions often follow a heavy-tailed pattern \cite{650143,workoladcharacterisation,crovella1998heavy,williams2005web}, complicating load balancing. In some cases, full recovery of results may not be feasible, and partial recovery is sufficient, as in coding theory \cite{DBLP:journals/corr/BalajiK15,KORHONEN2009229}. This has been applied to distributed computing, where approximations of the exact results are acceptable \cite{ozfatura2021coded, sarmasarkar2022gradient,10.1145/3366700}. Our work follows this paradigm, where the master aims to recover a large fraction of jobs in a noisy environment.


%% file: premiliminary-isit.tex
\section{Preliminaries and Notation}
{\label{setting_section}}




Given a set of $n$ jobs and $c$ servers, we study various assignments of jobs to different servers by the master server. More formally, let us denote the $n$ jobs by $\mathcal{A} = \{a_1,\ldots,a_n\}$ and $c$ servers by $\mathcal{S} = \{s_1,s_2,\ldots,s_c\}$. Any assignment ($D$) of jobs in $\mathcal{A}$ to servers in $\mathcal{S}$, can be  equivalently represented by a bipartite graph $\mathcal{G}_D$ where the nodes denote the jobs and the servers. The edges of the graph exist between nodes representing job $a_i$ and server $s_j$ if job $a_i$ is assigned to server $s_j$ with $A_{{D}} \in \{0,1\}^{n \times c}$ denoting its (bi)adjacency matrix i.e. $A_D[i,j] = 1$ iff job $a_i$ is assigned to server $s_j$.







\begin{nameddefinition}{Definition}{\label{definition_balanced}}
(Balanced $(n,k,r,c)$ assignment): Given a set of $n$ jobs and $c$ servers, we call an assignment scheme of jobs a balanced $(n,k,r,c)$ assignment if the following conditions are satisfied.
\begin{itemize}
    \item Each server is assigned precisely $k$ distinct jobs to compute.
    \item Each job is assigned to precisely $r$ distinct servers.
\end{itemize}
Note that this assignment scheme ensures that $n \times r=k \times c$.
\end{nameddefinition}



We now look at a balanced $(9,3,2,6)$ assignment scheme.


\begin{example}{\rm
We describe a balanced assignment scheme with 9 jobs $\{a_1,a_2,\ldots,a_9\}$ and 6 servers $\{s_1,s_2,\ldots,s_6\}$ in Table \ref{balanced_assn}. Note that each job is assigned to precisely $2$ servers and each server has exactly 3 jobs to compute. The assignment scheme is motivated from a cyclic assignment scheme.
}\end{example}

\begin{table}[!h]
    \begin{center}
{
	\begin{tabular}{|c|c|c|c| c|c|c|c| } 
		\hline
	{{\diagbox[width = 5 em]{Jobs}{Servers}}}& $s_1$ & $s_2$ & $s_3$ & $s_4$ & $s_5$ & $s_6$ \\
    \hline
    $a_1$ & 1 & 1 & $\text{}$ & $\text{}$ & $\text{}$ & $\text{}$ \\
    \hline
    $a_2$ & $\text{}$ & 1 & 1 & $\text{}$ & $\text{}$ & $\text{}$\\
        \hline
     $a_3$ & $\text{}$ & $\text{}$ & 1 & 1 & $\text{}$ & $\text{}$\\
        \hline
    
    $a_4$ & $\text{}$ & $\text{}$ & $\text{}$ & 1 & 1 & $\text{}$\\
    \hline
    $a_5$ & $\text{}$ & $\text{}$ & $\text{}$ & $\text{}$ & 1 & 1\\
        \hline
     $a_6$ & 1& $\text{}$ & $\text{}$ & $\text{}$ & $\text{}$ & 1 \\
        \hline
    
    $a_7$ & 1 & 1 & $\text{}$ & $\text{}$ & $\text{}$ & $\text{}$\\
    \hline
    $a_8$ & $\text{}$ & $\text{}$ & 1 & 1 & $\text{}$ & $\text{}$ \\
        \hline
     $a_9$ & $\text{}$ & $\text{}$ & $\text{}$ & $\text{}$ & 1 & 1 \\
        \hline
        
	\end{tabular}
 }
 \end{center}
	\vspace{1 em}
	\caption{Assignment of jobs to various servers in a balanced $(9,3,2,6)$ assignment scheme}
	\label{balanced_assn}
\end{table}

%% file: mean-and-variance-isit.tex
\section{The Mean and the Variance}{\label{setup_section}}

We consider the number of distinct jobs \(d\) received by the master when a subset of \(x\) servers successfully communicate with it. Any subset of servers \(\hat{S} \subseteq \mathcal{S}\) with \(|\hat{S}| = x\) is equally likely to be the set of communicators. If \(\hat{S}\) represents the subset of servers that communicate, the number of distinct jobs received is \(d = |\cup_{j \in \hat{S}} \text{supp}(A_D[:,j])|\), where \(\text{supp}(v)\) denotes the indices of non-zero entries in vector \(v\).

Next, we define the uniform distribution over all subsets of cardinality \(x\) from \(\mathcal{S}\), denoted by \(\mathfrak{D}_{\mathcal{S},x}\).
For a given assignment ${D}$ of jobs to servers, we denote the expectation and the variance of the number of distinct completed jobs received by the master when any subset of $x$ servers is able to communicate with master uniformly at random by $\mathbbm{E}_{{D},x}[d]$ and $\sigma_{{D},x}[d]$ respectively. Although the assignment scheme $D$ itself is deterministic, there is still randomness in the set of straggling workers.





Theorem~\ref{all_mean} states that the expectation on the number of distinct jobs $\mathbb{E}_{D,x}[d]$ is the same for every balanced $(n,k,r,c)$ assignment which is a function of $n,k,r,c$ and $x$ and is independent of the balanced assignment $D$. Throughout the paper, $\mathfrak{n}^D_{i,\hat{S}} \in \{0,1,\ldots,r\}$ denotes the number of servers in $\hat{S}$ to which job $a_i$ is assigned under the assignment scheme $D$ and $\mathbbm{N}$ denotes the set of positive integers including 0.  



\begin{theorem}{\label{all_mean}}
 Consider any balanced $(n,k,r,c)$ assignment $D$. The expectation of the number of distinct completed jobs $d$ received by the master when any subset of cardinality $x$ of the set of servers $\mathcal{S}$ is able to communicate with the master with equal probability is the same for every balanced $(n,k,r,c)$ assignment $D$ and is given by 
 
 \begin{equation}
  \mathbbm{E}_{{D},x}[d]= n\cdot\left(1-\frac{{c-r \choose x}}{{c \choose x}}\right) 
 \end{equation}

\end{theorem}

A detailed proof is presented in Appendix A. 


\begin{proof}[Proof Sketch]



We consider the number of distinct jobs \(d\) received by the master when a subset \(\hat{S}\) (with \(|\hat{S}| = x\)) communicates with the master:
\[
d = \left|\bigcup_{j \in \hat{S}}\text{supp}(A_D[:,j])\right| = k \times x - \sum_{i=1}^{n} (\mathfrak{n}^D_{i,\hat{S}} - 1) \mathbbm{1}_{\mathfrak{n}^D_{i,\hat{S}} > 1}
\]

The term \(\sum_{i=1}^{n} (\mathfrak{n}^D_{i,\hat{S}} - 1) \mathbbm{1}_{\mathfrak{n}^D_{i,\hat{S}} > 1}\) excludes jobs received multiple times from different servers in \(\hat{S}\). Next, we compute the expected number of distinct jobs:

\[
\mathbb{E}_{D,x}[d] = k \times x - \frac{n \sum_{\substack{\hat{S} \subset \mathcal{S}; \\ |\hat{S}| = x}} (\mathfrak{n}^D_{i,\hat{S}} - 1) \mathbbm{1}_{\mathfrak{n}^D_{i,\hat{S}} > 1}}{{c \choose x}}
\]

We now show that the sum \(\sum_{\hat{S} \subset \mathcal{S}, |\hat{S}| = x} (\mathfrak{n}^D_{i,\hat{S}} - 1) \mathbbm{1}_{\mathfrak{n}^D_{i,\hat{S}} > 1}\) is the same for all \(i\) in a balanced \((n, k, r, c)\) assignment. This is computed by counting the number of subsets \(\hat{S} \subset \mathcal{S}\) of cardinality \(x\) where job \(a_i\) appears more than once. Thus,


\[
\sum_{\substack{\hat{S} \subset \mathcal{S}; \\ |\hat{S}| = x}} (\mathfrak{n}^D_{i,\hat{S}} - 1) \mathbbm{1}_{\mathfrak{n}^D_{i,\hat{S}} > 1} = \sum_{t=1}^{r-1} t {r \choose (t+1)} {(c-r) \choose (x-t-1)}
\]

On substitution, we obtain the desired result.

\end{proof}

A few comments are in order here. Note that for $x=1$, the expectation (as expected) is precisely $k$. Observe that for $x>c-r$, the expectation goes to $n$ since the master obtains at least one copy of every job $a_i \in \mathcal{A}$. This is because each job is assigned to exactly 
$r$ servers, ensuring full coverage of all jobs when more than $
c-r$ servers are non-straggling.



We now calculate the variance for the number of distinct jobs $d$ received at the master for any balanced $(n,k,r,c)$ job assignment $D$. From the previous paragraph, it is clear that $\sigma_{D,1}[d] = 0$ for the cases $x = 1$
and $x>c-r$ (for sake of completeness, we calculate this in Corollary 1 in Appendix I). For calculating the variance on the number of distinct jobs $d$ received by the master, observe 
\begin{equation*}
 \sigma_{D,x}(d) = \sigma_{D,x} \left(\sum\limits_{i} (\mathfrak{n}^D_{i,\hat{S}}-1) \mathbbm{1}_{\mathfrak{n}^D_{i,\hat{S}}>1}\right)   
\end{equation*}

Using a similar technique (detailed calculation in Appendix B
) as described above, we obtain 
\begin{align}{\label{sigma_expression_final}}
\sigma_{D,x} (d)
        =  & \frac{\sum\limits_{1\leq i < j\leq n} \sum\limits_{\substack{\hat{S} \subset \mathcal{S};\\|\hat{S}|=x}} (\mathfrak{n}^D_{i,\hat{S}}-1) \mathbbm{1}_{\mathfrak{n}^D_{i,\hat{S}}>1} (\mathfrak{n}^D_{j,\hat{S}}-1) \mathbbm{1}_{\mathfrak{n}^D_{j,\hat{S}}>1}}{{{c \choose x}}}\nonumber\\
        & + \frac{T_2(n,r,c,x)}{{{c \choose x}}}- (T_1(n,r,c,x))^2 
\end{align}

where, $T_1(.)$ and $T_2(.)$ is defined in Appendix B. 

We analyze the term $\sum_{1 \leq i < j \leq n} \sum_{\substack{\hat{S} \subset \mathcal{S}; \\ |\hat{S}| = x}} (\mathfrak{n}^D_{i,\hat{S}} - 1) \mathbbm{1}_{\mathfrak{n}^D_{i,\hat{S}} > 1} (\mathfrak{n}^D_{j,\hat{S}} - 1) \mathbbm{1}_{\mathfrak{n}^D_{j,\hat{S}} > 1}$
where for a pair of jobs \( a_i \) and \( a_j \), if they appear \((\alpha + 1)\) and \((\beta + 1)\) times respectively in some subset \(\hat{S}\), their contribution to the sum is \(\alpha \beta\). Observe that the contribution of a pair of servers $(i,j)$ to the sum depends on precisely the number of servers (call it $m$) that are assigned both the jobs. Thus, we define



\[
g(m, x) = \sum_{\substack{\hat{S} \subset \mathcal{S}; \\ |\hat{S}| = x}} (\mathfrak{n}^D_{i,\hat{S}} - 1) \mathbbm{1}_{\mathfrak{n}^D_{i,\hat{S}} > 1} (\mathfrak{n}^D_{j,\hat{S}} - 1) \mathbbm{1}_{\mathfrak{n}^D_{j,\hat{S}} > 1}
\]

to represent the contribution from a pair of jobs assigned together to exactly \( m \) servers. We show that \( g(m, x) \) depends only on \( c \), \( r \), \( m \), and \( x \), and derive its value for \( g(0, x) \) in Lemma \ref{g_p_0_lemma} and a recursion for \( g(m, x) \) in Lemma \ref{g_p_diff_lemma}.

\begin{lemma}{\label{g_p_0_lemma}}
    For a balanced $(n,k,r,c)$ assignment, the value of $g(0,x)$ is given by 
    \begin{align}{\label{g_0_exp}}
    g(0,x) =  r^2 {{c-2 \choose x-2}} & - 2r {{c-1 \choose x-1}}  + {{c \choose x}} - 2{{c-r \choose x}} \nonumber\\
    & + 2 r{{c-r-1 \choose x-1}} + {c-2r \choose x}
\end{align}
\end{lemma}

\begin{lemma}{\label{g_p_diff_lemma}}
For a balanced $(n,k,r,c)$ assignment, the values for $g(m,x)$ are related in the following fashion 
\begin{align}{\label{g_p_diff}}
    g(m+1,x) - g(m,x) = {c-2 \choose x-1} & - 2{c-r-1 \choose x-1} \nonumber\\
    & + {c-2r+m \choose x-1}
\end{align}
\end{lemma}

We prove these Lemmas in Appendix F 
  and G 
respectively using a careful application of techniques from combinatorics. 


Note that all the expressions for $g(m,x)$ depends on the values of $c,r,m$ and $x$ and is therefore independent of which balanced $(n,k,r,c)$ assignment $D$ we choose. Observe further that the expression for $g(m+1,x)-g(m,x)$ is an increasing function of $m$ when $x$ is fixed. We now define $\mathfrak{m}^D(m)$ as the number of distinct pairs of jobs $(a_i,a_j)$ with $1 \le i < j \le n$ that are assigned together to precisely $m$ servers in the balanced $(n,k,r,c)$ assignment $D$. Further, since the total number of job pairs is ${{n \choose 2}}$ and since the number of pairs of jobs assigned to a fixed server $c_i$ is given by ${{k \choose 2}}$, we obtain equations \eqref{proporty_m_1} and \eqref{property_m_2}.



\noindent\begin{minipage}{0.49\linewidth}
    \begin{equation}
     \sum_{m=0}^r \mathfrak{m}^D(m) = {n \choose 2}
    \label{proporty_m_1}
    \end{equation}
\end{minipage}
\hfill
\begin{minipage}{0.48\linewidth}
    \begin{equation}
    \sum_{m=0}^r m\mathfrak{m}^D(m) = c{k \choose 2}
    \label{property_m_2}
    \end{equation}
\end{minipage}


Thus, we have shown that while the mean of the number of received jobs $d$ is the same for all balanced $(n,k,r,c)$ assignments, the variance of $d$ is dependent on the frequency distribution of job pairs assigned to the same server ($\mathfrak{m}^D(.)$). Observe that the variance of the number of received jobs scales linearly with $\sum\limits_{m=0}^{r}\mathfrak{m}^D(m)g(m,x)$ and we state a formal result in equation (23) in Appendix B. 

%% file: extreme-varaince-results-isit.tex
\section{Results on extreme variance based on $[\mathfrak{m}^D(m)]_{m=0}^{r}$}\label{extreme_var_section}

To explore the range of variances that balanced $(n,k,r,c)$ job assignments attain, we define a shape vector $h_D$.



\begin{definition}
Given a balanced $(n,k,r,c)$ job assignment $D$, we define a shape vector $h_D \in \mathbb{N}^{r+1}$  associated to the assignment $D$ as
\begin{equation}{\label{vec_h_defn}}
    h_D = [\mathfrak{m}^D(0), \mathfrak{m}^D(1), \ldots \mathfrak{m}^D(r)]^T
\end{equation}
\end{definition}

Two distinct balanced $(n,k,r,c)$ job assignments $D_1, D_2$ would have the same mean and variance if and only if the corresponding shape vectors $h_{D_1}$ and $h_{D_2}$ are the same. Using equations \eqref{proporty_m_1} and \eqref{property_m_2}, we have 
\begin{equation}{\label{defn_of_A}}
    H h_D = \begin{pmatrix} {n \choose 2} \\ c{k \choose 2} \end{pmatrix} \text{where}, H = \begin{bmatrix}
        1 & 1 & 1 & \cdots & 1 \\
        0 & 1 & 2 & \cdots & r
    \end{bmatrix}
    \end{equation}  


    Thus, two possible shape vectors differ by a vector in the kernel of the matrix $H$. 
    

\begin{definition}{\label{minimal_compact}}
A balanced $(n,k,r,c)$ assignment $D$ is compact if the shape vector $h_D$ has at most two non-zero elements.   
\end{definition}

When the shape vector \( h_D \) may have only one nonzero entry, meaning every job pair is assigned to \( m = \frac{ck(k-1)}{n(n-1)} = r\frac{k-1}{n-1} \) servers. We examine the policies achieving minimal and maximal variance of \( d \) in Sections \ref{sec:minimal_variance} and \ref{sec:maximal_variance}.



\subsection{Minimal variance of $d$}{\label{sec:minimal_variance}}

\begin{definition}{\label{proximally_minimal_compact}}
    A balanced assignment $D$ is proximally compact if the shape vector $h_D$ has either exactly one nonzero entry or has exactly two consecutive nonzero entries.    
\end{definition}

\begin{namedlemma}{Lemma}{\label{lemma_on_l_proximally_minimal}}
    For proximally compact $(n,k,r,c)$ assignment $D$, we have $\ell=\Bigl\lfloor\frac{r(k-1)}{n-1}\Bigr\rfloor$ where $\ell$ denotes the index of the smallest non-zero entry in the shape vector $h_D$.
\end{namedlemma}





The proof of this lemma follows since $\mathfrak{m}^D(m)$ is zero for $m \neq l,l+1$ we defer it to Appendix C. 
Given $n,k,r$ and $c$, it is not clear whether a balanced $(n,k,r,c)$ job assignment which is proximally compact exists. When the shape vector has only one nonzero entry corresponds to the case of balanced incomplete block designs (BIBD) \cite{doi:10.1111/j.1469-1809.1939.tb02219.x,10.5555-1202540}. BIBDs is a well studied subject and we redefine it below.


\begin{nameddefinition}{Defintion}
    (BIBD $(v,b,r,k,\lambda)$ scheme as in \cite{10.5555-1202540}) - A balanced incomplete block design (BIBD) is a pair $(V, B)$ where V is a $v$-set and B
is a collection of $b$ $k$-sized subsets of $V$ (blocks) such that each element of $V$ is contained
in exactly $r$ blocks and any 2-subset of V is contained in exactly $\lambda$ blocks.
\end{nameddefinition}

Note that we can associate the set $V$ to the set of jobs $\mathcal{A}$. Thus $v$ is the same as $n$ that we have employed so far. Each $k$ sized subset of $V$ (or $\mathcal{A}$) can be identified to the set of jobs assigned to a server. The number $r$ has the same interpretation as in our case. Since $B$ is a collection of $b$  $k$-sized sets, we can think of the number of servers $c$ being equal to $b$, thus $\lambda = \frac{r(k-1)}{n-1}$. $(n,k,r,c)$ in our case is identical to $(v,b,r,k,\lambda)$ quoted in the definition of BIBD above.  


The existence of BIBDs is an open question and the famous Bruck-Ryser-Chowla theorem in \cite{sprott1955} gives some necessary conditions. Proximally compact assignments may be thought of as a generalization of BIBDs that do not insist on a unique number $\lambda$, that represents the number of servers to be shared by every pair of jobs. We now provide an example of a proximally compact assignment scheme that is not a BIBD. 

\begin{example}{\label{proximally_compact_Example}} {\rm
    Consider balanced $(9,3,3,9)$ assignment schemes. In this case, $\frac{r(k-1)}{n-1} = \frac{3}{4}$ and so there is no BIBD possible. Further, $\ell = 0$ and the corresponding shape vector for a possible proximally compact assignment should be $h_D = [9, 
 27, 0, 0]^T$. We display an assignment scheme in Table \ref{proximally_minimally_compact_design} whose shape vector is indeed $h_D$. Note that in this scheme, 27 pairs of jobs are assigned together to a server once and there are 9 pairs of jobs that were never assigned together.}    

\end{example}

\begin{table}[!h]
    \centering
{
	\begin{tabular}{|c|c|c|c|c|c|c|c|c|c| } 
		\hline
	{{\diagbox[width = 5 em]{Jobs}{Servers}}}& $s_1$ & $s_2$ & $s_3$ & $s_4$ & $s_5$ & $s_6$ &$s_7$ & $s_8$ & $s_9$ \\
    \hline
    $a_1$ & 1 & 1 & 1 & $\text{}$ & $\text{}$ & $\text{}$ & $\text{}$ & $\text{}$ & $\text{}$ \\
    \hline
    $a_2$ & $\text{}$ & $\text{}$ & $\text{}$ & 1 & 1 & 1 & $\text{}$ & $\text{}$ & $\text{}$ \\
    \hline
     $a_3$ & $\text{}$ & $\text{}$ & $\text{}$ & $\text{}$ & $\text{}$ & $\text{}$ & 1 & 1 & 1 \\
        \hline
    $a_4$ & 1 & $\text{}$ & $\text{}$ & 1 & $\text{}$ & $\text{}$ & 1 & $\text{}$ & $\text{}$\\
    \hline
    $a_5$ & $\text{}$ & 1 & $\text{}$ & $\text{}$ & 1 & $\text{}$ & $\text{}$ & 1 & $\text{}$\\
    \hline
     $a_6$ & $\text{}$ & $\text{}$ & 1 & $\text{}$ & $\text{}$ & 1 & $\text{}$ & $\text{}$ & 1\\
    \hline
    
    $a_7$ & 1 &  $\text{}$ & $\text{}$ & $\text{}$ & 1 & $\text{}$ & $\text{}$ & $\text{}$ & 1 \\
    \hline
    $a_8$ & $\text{}$ & 1 & $\text{}$ & $\text{}$ & $\text{}$ & 1 & 1 & $\text{}$ & $\text{}$\\
    \hline
    $a_9$ & $\text{}$ & $\text{}$ & 1 & 1 & $\text{}$ & $\text{}$ & $\text{}$ & 1 & $\text{}$ \\
        \hline

	\end{tabular}
 }
	\vspace{1 em}
	\caption{Assignment of jobs to servers in a proximally minimally compact $(9,3,3,9)$ assignment scheme}
        \label{proximally_minimally_compact_design}
	
\end{table}


\begin{theorem}{\label{min_var}}
 If a proximally compact balanced $(n,k,r,c)$ job assignment exists, then it has the least variance amongst all balanced $(n,k,r,c)$ job assignments. 
\end{theorem}

A detailed proof is presented in Appendix D. 

\begin{proof}[Proof Sketch]
Let $h_D$ be the shape vector corresponding to the proximally compact balanced $(n,k,r,c)$ job assignment and let $h_D+v$ denote a shape vector for any balances assignment for $v \in \ker H$ with matrix $H$ as defined in \eqref{defn_of_A}. Since, only the term $\sum\limits_{m=0}^{r}\mathfrak{m}^D(m)g(m,x)$ varies amongst the different balanced $(n,k,r,c)$ assignments, it is sufficient to prove $\sum\limits_{m=0}^{r}v(m+1)g(m,x) \ge 0$ to conclude that the proximally compact balanced $(n,k,r,c)$ assignment has the least variance.


We therefore first characterize $v \in \ker H$ that may appear from some balanced $(n,k,r,c)$ assignment. As both the shape vectors $h_D, h_D+v \in \mathbb{N}^{r+1}$, therefore $v(i) \ge 0$ for all $i \neq \ell+1, \ell + 2$. Further, as $v \in \ker H$, therefore $\sum\limits_{i=1}^{r+1} v(i) = 0$ and so if $v$ is a nonzero vector, then at least one of $v(\ell+1), v(\ell+2)$ must be a negative integer. As $v \in \ker H$, therefore $v$ can be expressed in terms of the basis vectors $\left\{h_i\right\}$ of $\text{ker}(H)$. Let $v = \sum\limits_{i=1}^{r-1}\alpha_ih_i$ and we show that all $\alpha_i \in \mathbb{N}$. We conclude the proof by showing that $\sum\limits_{m=0}^{r} v(m+1)g(m,x) \geq 0$. 

\end{proof}

This theorem guarantees that if a proximally compact balanced \((n,k,r,c)\) assignment exists, it has the least variance. From equation (28) in Appendix D, for \(x = 1\), the contribution from each permissible \(v \in \ker H\) is zero, meaning all shape vectors give the same (zero) variance, consistent with our earlier result. Similarly, for \(x > c - r\), every \({{c-2r+i-1}\choose {x-1}}\) vanishes, and thus, we conclude that every balanced \((n,k,r,c)\) assignment has the same (zero) variance for \(x > c - r\). We present an example of  $(n,k,r,c)$ without a proximally compact assignment below.

\begin{example}{\rm
Consider balanced $(10,5,4,8)$ assignments. Here $\frac{r(k-1)}{n-1} = \frac{16}{9}$ and so the shape vector must have at least two nonzero entries. Moreover, $\ell = 1$ and the corresponding shape vector for a possible proximally compact assignment should be $h_D = [0, 10, 35, 0, 0]^T$. We now show that an assignment with the shape vector $h_D$ does not exist. 


As every job is assigned to $r=4$ servers, every job is involved in $r(k-1) = 16$ job pairs. Consider the job $a_1$ and let $x$ be the number of other jobs with whom $a_1$ shares only one server. So $9-x$ jobs share two servers each with $a_1$. Clearly $2(9-x) + x = 16$ which implies that $x = 2$. As $a_1$ was an arbitrary choice, we can conclude that every job shares one server with two other jobs and shares two servers with the other $7$ jobs. 

Let $\mathfrak{G} = \left\{a_1, a_2, a_3, \cdots a_i, a_1 \right\}$ be a cycle of jobs such that each job shares only one server with the jobs that are its predecessor and successor in $\mathfrak{G}$. Since there cannot exist a cycle of size 3 i.e. there cannot exist jobs $a,b$ and $c$ such that pairs $(a,b)$, $(b,c)$ and $(c,a)$ share a server each, one can now argue that the only permissible lengths of these cycles can be either 5 or 10. Then some more work proves that it is not possible to have an assignment with $h_D= [0, 10, 35, 0, 0]^T$  The actual balanced $(10,5,4,8)$ assignments that have shape vectors closest to $h_D$ have shape vectors $[1, 8, 36, 0, 0]^T$ and $[0, 12, 31, 2, 0]^T$ respectively.  

    }
\end{example}


\subsection{Maximal variance of $d$}{\label{sec:maximal_variance}}

\begin{definition}{\label{strectched_minimal_compact}}
    A balanced assignment $D$ is stretched compact if the shape vector $h_D$ has non-zero elements only in the first and the last entries.
\end{definition}

If only the first and last entries of the shape vector $h_D$ are nonzero, then by \eqref{defn_of_A} it is clear that the last entry of the shape vector is $\frac{c}{r}{k \choose 2} = \frac{n(k-1)}{2}$ and therefore the first entry of the shape vector is $\frac{n(n-k)}{2}$. Thus, if $n$ is an odd number and $k$ is even, there is no possibility of existence of a stretched compact $(n,k,r,c)$ assignment. Even otherwise the shape vector having integer entries does not guarantee the existence of stretched compact $(n,k,r,c)$ assignment.

\begin{theorem}{\label{max_var}}
 If a stretched compact balanced $(n,k,r,c)$ job assignment exists, then it has the largest variance amongst all balanced $(n,k,r,c)$ job assignments. 
\end{theorem}

 The proof is deferred to Appendix E. 

\begin{example}{\rm
Let us revisit the earlier example of $(10,5,4,8)$ assignments. It is clear that a shape vector corresponding to stretched compact assignment is permissible, namely $h_D= [25, 0, 0, 0, 20]^T$. Such an assignment is indeed possible. Divide the $10$ jobs into two sets of $5$ jobs. Assign each of these sets of jobs to $4$ servers. That results in a total of $20$ pairs sharing $4$ servers each and the rest $25$ pairs consisting of a job each from the two sets sharing no servers. This repetition assignment is a stretched compact assignment and therefore has the largest variance amongst all possible balanced $(10,5,4,8)$ assignments.  
}
\end{example}

If \( k \) divides \( n \), the jobs can be subdivided into \( \frac{n}{k} \) groups of \( k \) jobs each. Each of these groups are repeated at \( r \) servers, covering \( c = \frac{n}{k}r \) servers. The number of job pairs that appear together \( r \) times is \( \frac{n(k-1)}{2} \). This setup corresponds to a stretched compact assignment with the largest variance among all balanced \( (n, k, r, c) \) assignments.





\remove{

\begin{namedclaim}{Claim}{\label{Some claims om g(p,x)}}
	Recall the definition of $g(p,x)$ as defined in Equation \eqref{g_defn}. Then, the following can be said (for every $c-2r+1\geq x>1$):
	\begin{itemize}
		\item $\frac{g(m+k_1,x)-g(m,x))}{k_1}>\frac{g(m,x))-g(m-k_2,x)}{k_2}$ $\forall$ $k_1,k_2 \in \mathbbm{N}$.
		\item $\frac{g(m+k_1+v,x)-g(m+v,x)}{k_1}>\frac{g(m+u,x))-g(m-k_2+u,x)}{k_2}$ $\forall$ $k_1,k_2 \in \mathbbm{N}$ with $(k_1,k_2)\neq (1,1)$ and $(u,v) \in \{0,1\}^2$ with $u \neq v$.
	\end{itemize}

However, the inequalities may not strict for $x > c-2r+1$ but maybe {\color{red} attained with equality }
  
\end{namedclaim}

\begin{proof}

We first prove the strict inequalities assuming $x\leq c-2r+1$
 
	\begin{align}{\label{first_ineq}}
	  \frac{g(m+k_1,x)-g(m,x)}{k_1}
	  = \frac{\sum\limits_{i=0}^{i=k_1-1}(g(m+i+1,x)-g(m+i,x))}{k_1}
	 \overset{(a)}{\geq}& \frac{k_1 (g(m+1,x)-g(m,x)}{k_1}\nonumber\\        
     \geq & g(m+1,x)-g(m,x)
     \end{align}

Now consider,
	\begin{align}{\label{second_ineq}}
	\frac{g(m+k_1+1,x)-g(m+1,x)}{k_1}
	= \frac{\sum\limits_{i=1}^{i=k_1}(g(m+i+1,x)-g(m+i,x))}{k_1}
	\overset{(d)}{\geq}& \frac{k_1 (g(m+2,x)-g(m+1,x)}{k_1}\nonumber\\        
	> & g(m+2,x)-g(m+1,x)
	\end{align}
Note $(a)$ and $(d)$ follow since since $g(m+i+1,x)-g(m+i,x) > g(m+1,x)-g(m,x)$ $\forall i \in \mathbbm{N}$ as shown in in Claim \ref{g_p_diff_lemma} where 

$$g(m,x)-g(m-1,x)= \left[{{(c-2) \choose (x-1)}}-2{{(c-r-1) \choose (x-1)}}+ {{(c-2r+m-1)\choose (x-1)}}\right]$$

strictly increases with $m$ for every $c-2r+1>x>1$. Also observe that $(a)$ and $(d)$ would be tight inequalities when $k_1>1$.

Similarly,
  \begin{align}{\label{third_ineq}}
  \frac{g(m,x)-g(m-k_2,x)}{k_2}
  = \frac{\sum\limits_{i=0}^{i=k_2-1}(g(m-i,x)-g(m-i-1,x))}{k_2}
  \overset{(b)}{\leq}& \frac{k_2 (g(m,x)-g(m-1,x)}{k_2}\nonumber\\        
  \leq & g(m,x)-g(m-1,x)
  \end{align}

Similarly,
\begin{align}{\label{fourth_ineq}}
\frac{g(m+1,x)-g(m-k_2+1,x)}{k_2}
= \frac{\sum\limits_{i=0}^{i=k_2-1}(g(m-i+1,x)-g(m-i,x))}{k_2}
\overset{(c)}{\leq}& \frac{k_2 (g(m+1,x)-g(m,x)}{k_2}\nonumber\\        
\leq & g(m+1,x)-g(m,x)
\end{align}

Note $(b)$ and $(c)$ follow since $g(m-i,x)-g(m-i-1,x) \leq g(m,x)-g(m-1,x)$ $\forall i \in \mathbbm{N}$ as shown in Claim \ref{g_p_diff_lemma} where $g(m+1,x)-g(m,x)$ strictly increases with $m$ for every $c-2r+1>x>1$. Also observe that $(b)$ and $(c)$ would be tight inequalities when $k_2>1$.

Thus we can say from Equations \eqref{first_ineq},\eqref{second_ineq},\eqref{third_ineq},\eqref{fourth_ineq} and the fact that $g(m+1)-g(m)$ strictly increases with $m$, the following equations \eqref{result_1},\eqref{result_2},\eqref{result_3}.
\begin{equation}{\label{result_1}}
\frac{g(m+k_1,x)-g(m,x)}{k_1} > \frac{g(m,x)-g(m-k_2,x)}{k_2} \forall k_1,k_2 \in \mathbbm{N}
\end{equation} 

and 

\begin{equation}{\label{result_2}}
\frac{g(m+k_1,x)-g(m,x)}{k_1} > \frac{g(m+1,x)-g(m-k_2+1,x)}{k_2} \forall k_1,k_2 \in \mathbbm{N} \text{ with } (k_1,k_2) \neq (1,1).
\end{equation} 

\begin{equation}{\label{result_3}}
\frac{g(m+k_1+1,x)-g(m+1,x)}{k_1} > \frac{g(m+1,x)-g(m-k_2+1,x)}{k_2} \forall k_1,k_2 \in \mathbbm{N} \text{ with } (k_1,k_2) \neq (1,1).
\end{equation}

Thus the first and second inequalities in Claim \ref{Some claims om g(p,x)}are precisely the equations \eqref{result_1},\eqref{result_2} and \eqref{result_3}.


     

Observe that for $x> c-2r+1$, the function $g(m+1,x)-g(m,x)$ is not strictly increasing with $m$ as the last term ${{c-2r+m-1 \choose x-1}}$ maybe zero when $m$ equals $0$ or $1$ and thus the inequalities in the lemma maybe attained with equality.
\end{proof}

With this convexity result on $g(m,x)$, we now prove Theorem \ref{min_var_gen}.

\begin{proof}
	
%
%
%
%

Recall the definition of $\mathfrak{m}^D(m)$ from Equation \eqref{m_D_defn} which denotes the number of pairs of jobs that are assigned together to exactly $m$ servers with $D$ denoting any proximally minimally compact design.

Now we know that $\mathfrak{m}^{D}(m) = 0$ only for $m\neq l,l+1$ for any proximally minimally compact $(n,k,r,c)$ assignment scheme. This would ensure that $l= \floor{\frac{c.{{k \choose 2}}}{{{n \choose 2}}}}$ as shown in Lemma \ref{lemma_on_l_proximally_minimal}.


Let us consider another job assignment ${D}_1$ which is a balanced $(n,k,r,c)$ assignment scheme.
 


First observe that $\sum\limits_{m=0}^{r} \mathfrak{m}^D(m) = \sum\limits_{m=0}^{r} \mathfrak{m}^{D_1}(m) = {{n \choose 2}}$ which follows from Equation \eqref{proporty_m_1} and $\sum\limits_{m=0}^{r} m \times \mathfrak{m}^D(m) = \sum\limits_{m=0}^{r} m \times \mathfrak{m}^{D_1}(m)= c {{k \choose 2}}$ follows from Equation \eqref{property_m_2}.

We thus have 

\begin{equation}{\label{conditions_no_pairs}}
    \sum\limits_{m=0}^{r} \mathfrak{m}^D(m) = \sum\limits_{m=0}^{r} \mathfrak{m}^{D_1}(m) = {{n \choose 2}} \text{ and } \sum\limits_{m=0}^{r} m \times \mathfrak{m}^D(m) =\sum\limits_{m=0}^{r} m \times \mathfrak{m}^{D_1}(m)= c {{k \choose 2}}
\end{equation}

We now consider four different cases and in each of these cases, we show that the variance is the least for the assignment $D$. Also, we first consider the case where $x\leq c-2r+1$ to prove that least variance is uniquely attained by proximallyminimally compact assignments.

\begin{itemize}
	\item Case 1: $\mathfrak{m}^{D_1}(l) \geq \mathfrak{m}^{D}(l)$ but $\mathfrak{m}^{D_1}(l+1) \geq \mathfrak{m}^{D_1}(l+1)$

Note that $\mathfrak{m}^{D_1}(m) \geq \mathfrak{m}^D(m) \forall m \in [r]$ which follows since $\mathfrak{m}^D(m)=0 \forall m \neq l,l+1$. However, equation \eqref{conditions_no_pairs} would imply that $\mathfrak{m}^{D}(m)= \mathfrak{m}^{D_1}(m)$ $\forall$ $m \in [r]$ which would imply distribution $D_1$ has same variance of distinct jobs as that of distribution $D$ which follows from Equation \ref{var_sim_exp} in Theorem \ref{result_variance} as the variance $\sigma_{D,x}$ is just a function of $\mathfrak{m}^D(.)$ other than design parameters $n,k,r$ and $c$.

\item Case 2: $\mathfrak{m}^{D_1}(l) < \mathfrak{m}^{D}(l)$ but $\mathfrak{m}^{D_1}(l+1) \geq \mathfrak{m}^{D_1}(l+1)$

Let us denote $x_m=\mathfrak{m}^{D_1}(m)-\mathfrak{m}^{D}(m)$ $\forall m \in [0,r]$. Clearly, $x_m < 0$ only for $m=l$ as for every $m \neq l,l+1$ we have $\mathfrak{m}^D(m)=0$
	

Observe that equation \eqref{conditions_no_pairs} ensures that $\sum\limits_{m=0}^{r} x_m = \sum\limits_{m=0}^{r} m \times x_m = 0$.

 Let us denote $\sum\limits_{m=0}^{l-1} x_m =x$ and $\sum\limits_{m=l+1}^{r} x_m =y$. Since $\sum\limits_{m=0}^{r} x_m=0$, we can say that $x_l= -(x+y)$.
	
	Thus,
 
 \begin{equation}{\label{temp1}}
    \sum\limits_{m=0}^{r} m.x_m =0 \Leftrightarrow (x+y)l = \sum\limits_{m=l+1}^{r} m\times x_m+\sum\limits_{m<l}m\times x_m \Leftrightarrow \sum\limits_{m=l+1}^{r} x_m\times (m-l) = \sum\limits_{q=0}^{l-1} x_q\times (l-q)     
 \end{equation}

	However, we know from Claim \ref{Some claims om g(p,x)} that  

 \begin{equation}{\label{temp2}}
     \frac{g(m,x)-g(l,x)}{m-l} > \frac{g(l,x)-g(q,x)}{l-q} \forall m>l>q
 \end{equation}
	
These equations \eqref{temp1} and \eqref{temp2} and the fact that $x_m \geq 0$ for $m\neq l, m \in [0,r]$ above would imply:
	\begin{align}{\label{temp_eqn_variance_case1}}
    \sum\limits_{m=l+1}^{r} x_m (g(m,x)-g(l,x)) > \sum\limits_{q=0}^{l-1} x_q (g(l,x)-g(q,x))
	\overset{(c)} {\Leftrightarrow} & \sum\limits_{\substack{m \neq l\\ m \in [0,r]}} x_m.g(m,x) - \sum\limits_{\substack{m \neq l\\ m \in [0,r]}} x_m g(l,x)>0\nonumber \\
	\Leftrightarrow & \sum\limits_{\substack{m \neq l\\ m \in [0,r]}} x_m.g(m,x) + x_l.g(l,x) >0\nonumber\\
	{\Leftrightarrow} & \sum\limits_{m=0}^{r} x_m g(m,x)>0\nonumber \\
	\overset{(f)}{\Leftrightarrow} & \sum\limits_{m=0}^{r} \mathfrak{m}^{D_1}(m) g(m,x) > \sum\limits_{m=0}^{r} \mathfrak{m}^{D}(m) g(m,x)
	\end{align}
	
	Note $(c)$ follows since $\sum\limits_{m =0}^{r} x_m=0$ and $(f)$ follows from the fact that $x_m= \mathfrak{m}^{D_1}(m) - \mathfrak{m}^{D}(m)$

Now let us consider the numerator of the first term in $\sigma_{{D},x}(d)$ as in theorem \ref{var_sim_exp} which can be written as $2.\sum\limits_{m=0}^{r}\mathfrak{m}^{D}(m)g(m,x) + n.\left(\sum\limits_{t=1}^{r-1} t^2{r \choose (t+1)} {(c-r) \choose (x-t-1)}\right)$.

Thus the inequality proven in equation \eqref{temp_eqn_variance_case1} would imply that distribution $D_1$ has higher variance of number of distinct jobs received than that of distribution $D$. 

\item Case 3: $\mathfrak{m}^{D_1}(l+1) < \mathfrak{m}^{D}(l+1)$ but $\mathfrak{m}^{D_1}(l) \geq \mathfrak{m}^{D}(l)$

Note this can be proven in a very similar way as that of Case 2. The entire proof could be done for $l+1$ instead of $l$

\item Case 4: $\mathfrak{m}^{D_1}(l+1) < \mathfrak{m}^{D}(l+1)$ and  $\mathfrak{m}^{D_1}(l) < \mathfrak{m}^{D}(l)$

Let us denote $x_m=\mathfrak{m}^{D_1}(m)-\mathfrak{m}^{D}(m)$ $\forall m \in [0,r]$. Clearly, $x_m < 0$ only for $m=l,l+1$.
	
Now equation \eqref{conditions_no_pairs} ensures that $\sum\limits_{m=0}^{r} m\times x_m=0$ and $\sum\limits_{m=0}^{r} x_m=0$.

Let us denote $\sum\limits_{m=0}^{l-1}x_m =\alpha_1+\alpha_2$ and $\sum\limits_{m=l+2}^{r} x_m=\beta_1+\beta_2$. Since $\sum\limits_{m=0}^{r} x_m=0$, we can say that

\begin{equation}{\label{temp_10}}
    x_l= -(\alpha_1+\beta_1) \text{ and } x_{l+1}=-(\alpha_2+\beta_2) \text { for some }\alpha_1,\beta_1,\alpha_2,\beta_2 \in \mathbbm{N}
\end{equation}

Now define $[y_m]_{m=0}^{r}$ and $[z_m]_{m=0}^{r}$as follows below.

\begin{itemize}
    \item $y_m = x_m \times \frac{\alpha_1}{\alpha_1+\alpha_2}$ and $z_m = x_m \times \frac{\alpha_2}{\alpha_1+\alpha_2}$ if $0\leq m<l$ 
    \item $y_m=x_m\times \frac{\beta_1}{\beta_1+\beta_2}$ and $z_m=x_m\times \frac{\beta_2}{\beta_1+\beta_2}$
   for $l+1<m\leq r$
   \item $y_m= \alpha_1$ and $z_m = \alpha_2$ if $m=l$
   \item $y_m= \beta_1$ and $z_m = \beta_2$ if $m=l+1$
\end{itemize}


   Thus,
   
   $\sum\limits_{m=l+2}^{r} y_m=\beta_1$ ; $\sum\limits_{m=l+2}^{r} z_m=\beta_2$; $\sum\limits_{m=0}^{l-1}y_m=\alpha_1$ ; $\sum\limits_{m=0}^{l-1}z_m=\alpha_2$;	
   	
   	
	Thus $\sum\limits_{m=0}^{r} mx_m =0$ implies the following set of conditions.
	
	\begin{align}{\label{temp3}}
      & (\alpha_1+\beta_1)l+(\alpha_2+\beta_2)(l+1) = \sum\limits_{m=l+2}^{r}m\times x_m+\sum\limits_{m=0}^{l-1}m\times x_m \nonumber\\
    & \overset{(d)}{\Leftrightarrow} \left(\sum\limits_{m=l+2}^{r} y_m + \sum\limits_{m=0}^{l-1} y_m\right)l+\left(\sum\limits_{m=l+2}^{r} z_m + \sum\limits_{m=0}^{l-1} z_m\right)(l+1) = \sum\limits_{m=l+2}^{r}m\times (y_m+z_m)+\sum\limits_{m=0}^{l-1}m\times (y_m+z_m)\nonumber\\
     & {\Leftrightarrow} \sum\limits_{m=l+2}^{r} y_m(m-l)+ \sum\limits_{m=l+2}^{r} z_m(m-l-1) = \sum\limits_{q=0}^{l-1} y_q(l-q) + \sum\limits_{q=0}^{l-1} z_q(l+1-q)
	\end{align}
	
Note $(d)$ follows since 

i) $y_m+z_m = x_m \forall m \in [0,l-1] \cup [l+2,r]$, ii) $\sum\limits_{m=l+2}^{r} y_m + \sum\limits_{q=0}^{l-1} y_q = (\alpha_1+\beta_1)$ and \\
iii) $\sum\limits_{m=l+2}^{r} z_m + \sum\limits_{q=0}^{l-1} z_q = (\alpha_2+\beta_2)$ \\
	
	Now we can say from Claim \ref{Some claims om g(p,x)} that 
	
\begin{equation}{\label{temp4}}
    \frac{g(m,x)-g(t,x)}{m-t}> \frac{g(u,x)-g(q,x)}{u-q}
	\text{ }\forall m>l+1, q< l \text{ and } t,u \in \{l,l+1\}.
 \end{equation}

Thus we can say the following from equations \eqref{temp3} and \eqref{temp4} and the fact that $x_m,y_m,z_m>0$ for $m \in [0,r]; m \neq \{l,l+1\}$.
	
	\begin{align}{\label{temp_eqn_variance_case2}}
	 & \sum\limits_{m=l+2}^{r} y_m(g(m,x)-g(l,x))+  \sum\limits_{m=l+2}^{r} z_m(g(m,x)-g(l+1,x))\\ & > \sum\limits_{q=0}^{l-1} y_q(g(l,x)-g(q,x)) + \sum\limits_{q=0}^{l-1} z_q(g(l+1,x)-g(q,x))\\
	& \overset{(e)}{\Leftrightarrow} \sum\limits_{\substack{m \neq l,l+1\\ m \in [0,r]}} (y_m+z_m) g(m,x)+ (\alpha_1+\beta_1) g(l,x) + (\alpha_2+\beta_2) g(l+1,x)>0\\
	& \overset{(g)}{\Leftrightarrow} \sum\limits_{\substack{m \neq l,l+1\\ m \in [0,r]}} x_m.g(m,x)+ x_l.g(l,x) + x_{l+1}.g(l+1,x)>0\\
	& \Leftrightarrow \sum\limits_{m=0}^{r} x_m.g(m,x)>0\\
    & \overset{(h)}{\Leftrightarrow} \sum\limits_{m=0}^{r} \mathfrak{m}^{D_1}(m) g(m,x) > \sum\limits_{m=0}^{r} \mathfrak{m}^{D}(m) g(m,x)
	\end{align}
	
Note $(e)$ follows since $\sum\limits_{m=l+2}^{r} y_m + \sum\limits_{q=0}^{l-1} y_q = (\alpha_1+\beta_1)$ and $\sum\limits_{m=l+2}^{r} z_m + \sum\limits_{q=0}^{l-1} z_q = (\alpha_2+\beta_2)$ from Equation \eqref{temp_10} and the fact $y_m+z_m=x_m \forall m \in [0,l-1] \cup [l+2,r]$. \\ $(g)$ follows from the fact that $x_l=-(\alpha_1+\beta_1)$ and $x_{l+1}=-(\alpha_2+\beta_2)$ $(h)$ follows from the fact that that $x_m=\mathfrak{m}^{D_1}(m)-\mathfrak{m}^{D}(m)$ 	
	
Now let us consider the numerator of the first term in $\sigma_{{D},x}(d)$ as in theorem \ref{var_sim_exp} which can be written as $2.\sum\limits_{m=0}^{r}\mathfrak{m}^{D}(m)g(m,x) + n.\left(\sum\limits_{t=1}^{r-1} t^2{r \choose (t+1)} {(c-r) \choose (x-t-1)}\right)$.

Thus the inequality proven in the previous equation \eqref{temp_eqn_variance_case2} would imply that distribution $D_1$ has higher variance of number of distinct jobs received than that of distribution $D$.

Note that other than Case 1 in this proof (where we show $D_1$ is also a proximally minimally compact assignment), the distribution $D_1$ is not a proximally minimally compact assignment and we show strict optimality of proximally minimally compact assignment over any other balanced assignment for $x \leq c-2r+1$. However, if $x > c-2r+1$, the inequalities still hold but are not strict as it follows from Claim \ref{Some claims om g(p,x)}.

Combining these statements, we prove the strict optimality (attains least variance )of pairwise balanced designs for $x\leq c-2r+1$ and weak optimality for $x>c-2r+1$

\end{itemize}

\end{proof}

\begin{theorem}{\label{max_var_gen}}
    The largest variance on the number of distinct jobs received at the server for any $x \in [1,c-2r+1]$ is attained uniquely by the stretched minimally compact design (if exists) amongst all balanced $(n,k,r,c)$ assignments.

    However, for $x>c-2r+1$, the result on the largest variance still holds without guarantee of uniqueness.
\end{theorem}

We also prove this theorem using the convexity property of $g(m,x)$ as shown in Claim \ref{Some claims om g(p,x)}.

\begin{proof}
    Recall the definition of $\mathfrak{m}^D(m)$ from Equation \ref{m_D_defn} which denotes the number of pairs of jobs that are assigned together to exactly $m$ servers with $D$ denoting a stretched minimally compact design.

Now we know that $\mathfrak{m}^{D}(m) = 0$ only for $m\neq 0,r$.


Let us consider another job assignment ${D}_1$ which is a balanced $(n,k,r,c)$ assignment scheme.
 


First observe that $\sum\limits_{m=0}^{r} \mathfrak{m}^D(m) = \sum\limits_{m=0}^{r} \mathfrak{m}^{D_1}(m) = {{n \choose 2}}$ which follows from Equation \eqref{proporty_m_1} and $\sum\limits_{m=0}^{r} m \times \mathfrak{m}^D(m) = \sum\limits_{m=0}^{r} m \times \mathfrak{m}^{D_1}(m)= c {{k \choose 2}}$ follows from Equation \eqref{property_m_2}.

We thus have 

\begin{equation}{\label{conditions_no_pairs_repeat}}
    \sum\limits_{m=0}^{r} \mathfrak{m}^D(m) = \sum\limits_{m=0}^{r} \mathfrak{m}^{D_1}(m) = {{n \choose 2}} \text{ and } \sum\limits_{m=0}^{r} m \times \mathfrak{m}^D(m) =\sum\limits_{m=0}^{r} m \times \mathfrak{m}^{D_1}(m)= c {{k \choose 2}}
\end{equation}

We now consider three different cases and in each of these three cases, we show that the variance is the largest for the assignment $D$. Also, we first consider the case where $x\leq c-2r+1$ to prove that the largest variance is uniquely attained by pairwise heavy imbalanced designs.

\begin{itemize}
	\item Case 1: $\mathfrak{m}^{D_1}(r) \geq \mathfrak{m}^{D}(r)$

 Note that this implies that $\sum_{m=0}^{r} m \times \mathfrak{m}^{D_1}(m) \geq \sum_{m=0}^{r} m \times \mathfrak{m}^{D}(m)$
as $\mathfrak{m}^D(m)=0\forall m \in [0,r-1]$. Thus, equation \eqref{conditions_no_pairs_repeat} implies that $\mathfrak{m}^{D_1}(m)=\mathfrak{m}^{D}(m) \forall m \in [r]$ which would imply distribution $D_1$ has same variance of distinct jobs as that of distribution $D$ which follows from Equation \eqref{var_sim_exp} in Theorem \ref{result_variance} as the variance $\sigma_{D,x}$ is just a function of $\mathfrak{m}^D(.)$ other than design parameters $n,k,r$ and $c$.


\item Case 2: $\mathfrak{m}^{D_1}(r) < \mathfrak{m}^{D}(r)$ but $\mathfrak{m}^{D_1}(0) \geq \mathfrak{m}^{D}(0)$
	
Let us denote $x_m=\mathfrak{m}^{D_1}(m)-\mathfrak{m}^{D}(m)$ $\forall m \in [0,r]$. Clearly, $x_m < 0$ only for $m=r$ as for every $m \neq r$ we have $\mathfrak{m}^D(m)=0$.
	

Observe that equation \eqref{conditions_no_pairs} ensures that $\sum\limits_{m=0}^{r} x_m = \sum\limits_{m=0}^{r} m \times x_m = 0$.

 Let us denote $\sum\limits_{m=0}^{r-1} x_m =x$. Since $\sum\limits_{m=0}^{r} x_m=0$, we can say that $x_r= -x$.
	
	Thus,
 
 \begin{align}{\label{temp1_repeat}}
        \sum\limits_{m=0}^{r} x_m \times g(m,x) & = \sum\limits_{m=0}^{r-1} x_m\times g(m,x) - x \times g(r,x) \nonumber \\
        & = \sum\limits_{m=0}^{r-1} x_m \times(g(m,x)-g(r,x)) \overset{(a)}{\leq} 0.
 \end{align}

Note that $(a)$ follows from the fact that $x_m>0$ for all $m \in [0,r-1]$ and the fact that $g(m,x)$ is monotonic in $m$.

\item Case 3: $\mathfrak{m}^{D_1}(0) < \mathfrak{m}^{D}(0)$ and  $\mathfrak{m}^{D_1}(r) < \mathfrak{m}^{D}(r)$

Let us denote $x_m=\mathfrak{m}^{D_1}(m)-\mathfrak{m}^{D}(m)$ $\forall m \in [0,r]$. Clearly, $x_m < 0$ only for $m=0,r$.
	
Now equation \eqref{conditions_no_pairs_repeat} ensures that $\sum\limits_{m=0}^{r} m\times x_m=0$ and $\sum\limits_{m=0}^{r} x_m=0$ Observe that $x_0+x_r = - \sum\limits_{m=1}^{r-1}  x_m$


\remove{
\begin{equation}{\label{temp_10}}
    x_l= -(\alpha_1+\beta_1) \text{ and } x_{l+1}=-(\alpha_2+\beta_2) \text { for some }\alpha_1,\beta_1,\alpha_2,\beta_2 \in \mathbbm{N}
\end{equation}
}

Now define $[y_m]_{m=0}^{r}$ and $[z_m]_{m=0}^{r}$as follows below.

\begin{itemize}
    \item $y_m = x_m \times \frac{x_0}{x_0+x_r}$ and $z_m = x_m \times \frac{x_r}{x_0+x_r}$ if $0< m<r$ 
\end{itemize}
			
Also note that $y_m+z_m = x_m\forall m \in [1,r-1]$ and also  $\sum_{l=1}^{r-1} y_l = -x_0$ and $\sum_{l=1}^{r-1} z_l = -x_r$. Thus, $\sum\limits_{m=0}^{r} mx_m =0$ implies

\begin{equation}{\label{temp3_repeat}}
    \sum_{m=1}^{l-1} (r-m) z_m = \sum_{q=1}^{l-1} q \times y_q
\end{equation}

	Now we can say from Claim \ref{Some claims om g(p,x)} that 
	
\begin{equation}{\label{temp4_repeat}}
    \frac{g(r,x)-g(t,x)}{r-t}> \frac{g(t,x)-g(0,x)}{t-0}
	\text{ }\forall t \in [0,r].
 \end{equation}

Thus we can say the following from equations \eqref{temp3_repeat} and \eqref{temp4_repeat} and the fact that $x_m,y_m,z_m>0$ for $m \in [0,r];$
	
\begin{align}{\label{temp_eqn_variance_case2_repeat}}
	& \sum_{m=1}^{r-1} z_m (g(r,x)-g(m,x)) > \sum_{q=1}^{r-1} y_q (g(q,x)-g(0,x)) \nonumber\\
  \Leftrightarrow & (\sum_{m=1}^{r-1} z_m) g(r,x) - \sum_{m=1}^{r-1}( y_m+z_m) g(m,x) + (\sum_{m=0}^{r-1} y_m) g(0,x) > 0\nonumber\\
 \overset{(b)}{\Leftrightarrow} & -x_r g(r,x) - \sum_{m=1}^{r-1} x_m g(m,x) - x_0 g(0,x) > 0\nonumber\\
 {\Leftrightarrow} & \sum_{m=0}^{r} x_m g(m,r) < 0\nonumber\\
 {\Leftrightarrow} & \sum\limits_{m=0}^{r} \mathfrak{m}^{D_1}(m) g(m,x) < \sum\limits_{m=0}^{r} \mathfrak{m}^{D}(m) g(m,x)
\end{align}
	
Note $(b)$ follows from $y_m+z_m = x_m\forall m \in [1,r-1]$ and also  $\sum_{l=1}^{r-1} y_l = -x_0$ and $\sum_{l=1}^{r-1} z_l = -x_r$.

Now let us consider the numerator of the first term in $\sigma_{{D},x}(d)$ as in theorem \ref{var_sim_exp} which can be written as $2.\sum\limits_{m=0}^{r}\mathfrak{m}^{D}(m)g(m,x) + n.\left(\sum\limits_{t=1}^{r-1} t^2{r \choose (t+1)} {(c-r) \choose (x-t-1)}\right)$.

Thus the inequality proven in the previous equation \eqref{temp_eqn_variance_case2} would imply that assignment scheme $D_1$ has a smaller variance of the number of distinct jobs received at the master than that of assignment scheme $D$, thus proving our desired result.

Note that other than Case 1 in this proof (where we show $D_1$ is also a stretched minimally compact assignment), the distribution $D_1$ is not a stretched minimally compact design and we show strict optimality of stretched minimally compact design over any other balanced assignment for $x \leq c-2r+1$. However, if $x> c-2r+1$, the inequalities still hold but are not strict as it follows from Claim \ref{Some claims om g(p,x)}.

Combining these statements, we prove the strict optimality (attains largest variance ) of stretched minimally compact assignments for $x\leq c-2r+1$ and weak optimality for $x>c-2r+1$.

\end{itemize}

\end{proof}

}

%% file: random-return.tex
\section{Number of non-stragglers ($x$) is random} \label{sec_randomness_return_servers}

We now look at a scenario where each of the $c$ servers is independently and equally likely to communicate with the master with probability $p$. Note that under this setup, the distribution of the number of servers $x$ that could communicate is given by the binomial distribution $B(c,1-p)$. Observe that conditioned on $x$, every subset of $x$ servers is equally likely to be non-straggling. 

\begin{theorem}{\label{all_mean_sampled}}
Consider any balanced $(n,k,r,c)$ assignment $D$, where each server is independently and equally likely to straggle with probability $p$. The expectation of the number of distinct completed jobs $d$ received is $ \mathbbm{E}_{D}[d] = n- n(p)^r$.
 




 \end{theorem}

This result is shown in Appendix H from tower rule.


We can actually generalize some of our results in Theorem \ref{min_var} and \ref{max_var} for a more generalized setup where the number of servers that return is not unique but is sampled from some distribution $\mathcal{P}$. Formally, we study the setup where $x$ is sampled from a distribution $\mathcal{P}$ and conditioned on $x$, any subset of $x$ servers is equally likely to be the set of non-straggling servers. Choosing $\mathcal{P}$ as $B(c,1-p)$ implies that the every server is equally and independently likely to straggle with probability $p$.



\begin{theorem}{\label{max_min_var_sampled}}
    Let us consider $x \sim \mathcal{P}$. Conditioned on $x$, we study the setup where any set of $x$ servers is equally likely to communicate with the master. Then the proximally compact assignment (if it exists) and stretched compace assignment (if exists) attain the least and the largest variance respectively on the number of distinct jobs received at master amongst all balanced $(n,k,r,c)$ assignment schemes.
\end{theorem}

This result follows from the application of Eve's law on theorems \ref{min_var} and \ref{max_var} respectively and proven in Appendix H. 

%% file: conclusion.tex
\section{Conclusion}
We study the mean and variance of distinct jobs received by a server under various assignment schemes, showing that repetition coding and block designs achieve the largest and smallest variance, respectively. However, it is not always clear if such designs exist, with the Bruck-Ryser-Chowla theorem \cite{sprott1955} providing only necessary conditions. Future work could explore if slight modifications to these designs yield near-extremal variance, or extend the results to the \( t^{th} \) moment to investigate if \( t \)-designs \cite{10.5555-1202540} can achieve extremal results.

%% file: mean_theorem_proof.tex
\subsection{Proof of Theorem \ref{all_mean}}{\label{sec:all_mean_proof}}

\begin{theorem*}
 Consider any balanced $(n,k,r,c)$ assignment $D$. The expectation of the number of distinct completed jobs $d$ received by the master when any subset of cardinality $x$ of the set of servers $\mathcal{S}$ is able to communicate with the master with equal probability is the same for every balanced $(n,k,r,c)$ assignment $D$ and is given by 
 
 \begin{equation}
  \mathbbm{E}_{{D},x}[d]= n\cdot\left(1-\frac{{c-r \choose x}}{{c \choose x}}\right) 
 \end{equation}
\end{theorem*}

\begin{proof}



The number of distinct jobs $d$ received by the master when servers in a subset $\hat{S}$ (with $|\hat{S}| = x$) is able to communicate with the master is given by 

\begin{equation}{\label{constr_d}}
    d = \left|\bigcup_{j \in \hat{S}}\text{supp}(A_D[:,j])\right|= \left(k\times x- \sum\limits_{i=1}^{n} (\mathfrak{n}^D_{i,\hat{S}}-1) \mathbbm{1}_{\mathfrak{n}^D_{i,\hat{S}}>1}\right)
\end{equation}

Note that the term $\sum\limits_{i=1}^{n} (\mathfrak{n}^D_{i,\hat{S}}-1) \mathbbm{1}_{\mathfrak{n}^D_{i,\hat{S}}>1}$ excludes those jobs which have been received multiple times from various servers present in $\hat{S}$.

	\begin{align}{\label{expectation_simplification}}
	 \mathbb{E}_{D,x}[d]
	= \frac{\sum\limits_{\substack{\hat{S} \subset \mathcal{S};\\|\hat{S}|=x}} (k \times x- \sum\limits_{i=1}^{n} (\mathfrak{n}^D_{i,\hat{S}}-1) \mathbbm{1}_{\mathfrak{n}^D_{i,\hat{S}}>1})}{\sum\limits_{\substack{\hat{S} \subset \mathcal{S};\\|\hat{S}|=x}} 1}
	 {=} & k\times x - \frac{\sum\limits_{i=1}^{n} \sum\limits_{\substack{\hat{S} \subset \mathcal{S};\\|\hat{S}|=x}} (\mathfrak{n}^D_{i,\hat{S}}-1) \mathbbm{1}_{\mathfrak{n}^D_{i,\hat{S}}>1} }{{c \choose x}}\nonumber\\ {=} & k\times x - \frac{n\sum\limits_{\substack{\hat{S} \subset \mathcal{S};\\|\hat{S}|=x}} (\mathfrak{n}^D_{i,\hat{S}}-1) \mathbbm{1}_{\mathfrak{n}^D_{i,\hat{S}}>1} }{{c \choose x}}
	\end{align}

	
	Observe that for every job $a_i$ in a balanced $(n,k,r,c)$ assignment, the quantity $\sum\limits_{\hat{S} \subset \mathcal{S},|\hat{S}|=x} (\mathfrak{n}^D_{i,\hat{S}}-1) \mathbbm{1}_{\mathfrak{n}^D_{i,\hat{S}}>1}$ is the same, i.e., this summation is independent of $i$.  We now show that the quantity  $\sum\limits_{\hat{S} \subset \mathcal{S},|\hat{S}|=x} (\mathfrak{n}^D_{i,\hat{S}}-1) \mathbbm{1}_{\mathfrak{n}^D_{i,\hat{S}}>1}$ for any specified $x$, is the same for every balanced  $(n,k,r,c)$ distribution ${D}$. We compute this sum by counting the number of subsets $\hat{S} \subset \mathcal{S}$ of cardinality $x$ which additionally satisfies the constraint on $\mathfrak{n}^D_{i,\hat{S}}=t$ (i.e. job $a_i$ is present in exactly $t$ servers from $\hat{S}$) for every $t$ from $2$ to $r$ (as these cases deal with the job $a_i$ appearing more than once in the subset $\hat{S}$).

\begin{align}{\label{removal_comp_mean}}
	\sum\limits_{\substack{\hat{S} \subset \mathcal{S};\\|\hat{S}|=x}} (\mathfrak{n}^D_{i,\hat{S}}-1) \mathbbm{1}_{\mathfrak{n}^D_{i,\hat{S}}>1}
	{=} \sum\limits_{t=1}^{r-1} \sum\limits_{\substack{\hat{S} \subset \mathcal{S}; \\|\hat{S}|=x,\mathfrak{n}^D_{i,\hat{S}}=t+1}} t
	{=} \text{ }& \sum\limits_{t=1}^{r-1} t \sum\limits_{\substack{\hat{S} \subset \mathcal{S}; \\|\hat{S}|=x,\mathfrak{n}^D_{i,\hat{S}}=t+1}} 1 \nonumber\\
	\overset{(a)}{=} & \sum\limits_{t=1}^{r-1} t{r \choose (t+1)} {(c-r) \choose (x-t-1)} 
\end{align}



	
	

The last equality $(a)$ comes from counting the number of subsets $\hat{S} \subset \mathcal{S}$ of cardinality $x$ that contain precisely $t+1$ servers that were assigned the job $a_i$.
Consider the following binomial expressions	
	\begin{equation}{\label{binomial_1}}
	    ry(1+y)^{r-1} + 1 - {(1+y)^r} = \sum\limits_{t=0}^{r-1} t{{r \choose t+1}} y^{t+1}
	\end{equation}

        \begin{equation}{\label{binomial_2}}
            (1+y)^{c-r}=\sum\limits_{u=0}^{c-r} {{c-r \choose u}} y^u
        \end{equation}
 
Multiplying equations \eqref{binomial_1} and \eqref{binomial_2}, one observes that $\sum\limits_{t=1}^{r-1} t{r \choose (t+1)} {(c-r) \choose (x-t-1)}$ is precisely the coefficient of $y^{x}$ in $ry(1+y)^{c-1} + {(1+y)^{c-r}} - {(1+y)^c}$. Thus,

\begin{equation}{\label{temp_result_removal}}
    \sum\limits_{t=1}^{r-1} t{r \choose (t+1)} {(c-r) \choose (x-t-1)} = r \times {{c-1 \choose x-1}} + {{c-r \choose x}} - {{ c \choose x}}
\end{equation}
	
	


Combining equations \eqref{expectation_simplification},
\eqref{removal_comp_mean} and \eqref{temp_result_removal}, we get 

$$\mathbbm{E}_{D,x}[d]= k \times x - \frac{n\left(r \times {{c-1 \choose x-1}} + {{c-r \choose x}} - {{c \choose x}}\right)}{{{c \choose x}}}= n\left(1-\frac{{{c-r \choose x}}}{{{c \choose x}}}\right)$$

\end{proof}

%% file: variance_computation.tex
\subsection{Computation of Variance of distinct jobs $\sigma_{D,x}(d)$}{\label{sec:variance_comp}}

For calculating the variance on the number of distinct jobs $d$ received by the master, observe 
\begin{equation*}
 \sigma_{D,x}(d) =
	\sigma_{D,x}\left(k\times x- \sum\limits_{i} (\mathfrak{n}^D_{i,\hat{S}}-1) \mathbbm{1}_{\mathfrak{n}^D_{i,\hat{S}}>1}\right) = \sigma_{D,x} \left(\sum\limits_{i} (\mathfrak{n}^D_{i,\hat{S}}-1) \mathbbm{1}_{\mathfrak{n}^D_{i,\hat{S}}>1}\right)   
\end{equation*}
The above follows since $\sigma(t-X) = \sigma(X)$ where $t$ is a constant and $X$ is a random variable. We now make use of the definition $\text{var}(X) = \mathbbm{E}[X^2]- (\mathbbm{E}[X])^2$. Therefore,
\begin{align}{\label{sigma_expression1}}
	\sigma_{D,x} \left(\sum\limits_{i} (\mathfrak{n}^D_{i,\hat{S}}-1) \mathbbm{1}_{\mathfrak{n}^D_{i,\hat{S}}>1}\right)
	= & \frac{\sum\limits_{\substack{\hat{S} \subset \mathcal{S};\\|\hat{S}|=x}} {\left(\sum\limits_{i=1}^{n} (\mathfrak{n}^D_{i,\hat{S}}-1) \mathbbm{1}_{\mathfrak{n}^D_{i,\hat{S}}>1}\right)}^2}{\sum\limits_{\substack{\hat{S} \subset \mathcal{S};\\|\hat{S}|=x}} 1} - \left(\frac{\sum\limits_{\substack{\hat{S} \subset \mathcal{S}; |\hat{S}|=x}} {(\sum\limits_{i=1}^{n} (\mathfrak{n}^D_{i,\hat{S}}-1) \mathbbm{1}_{\mathfrak{n}^D_{i,\hat{S}}>1})}}{\sum\limits_{\substack{\hat{S} \subset \mathcal{S};\\|\hat{S}|=x}} 1}\right)^2\nonumber\\
	& \hspace{-2 em}\overset{(a)}{=} \frac{\sum\limits_{\substack{\hat{S} \subset \mathcal{S};\\|\hat{S}|=x}} (\sum\limits_{i=1}^{n} \sum\limits_{j=1}^{n} (\mathfrak{n}^D_{i,\hat{S}}-1) \mathbbm{1}_{\mathfrak{n}^D_{i,\hat{S}}>1} (\mathfrak{n}^D_{j,\hat{S}}-1) \mathbbm{1}_{\mathfrak{n}^D_{j,\hat{S}}>1})}{\sum\limits_{\substack{\hat{S} \subset \mathcal{S};\\|\hat{S}|=x}} 1} - \left(\frac{\sum\limits_{i=1}^{n} \sum\limits_{\substack{\hat{S} \subset \mathcal{S};\\ |\hat{S}|=x}} {((\mathfrak{n}^D_{i,\hat{S}}-1) \mathbbm{1}_{\mathfrak{n}^D_{i,\hat{S}}>1})}}{{c \choose x}}\right)^2\nonumber\\
	& \overset{(b)}{=} \frac{\sum\limits_{i=1}^{n} \sum\limits_{j=1}^{n} \sum\limits_{\substack{\hat{S} \subset \mathcal{S};\\|\hat{S}|=x}} (\mathfrak{n}^D_{i,\hat{S}}-1) \mathbbm{1}_{\mathfrak{n}^D_{i,\hat{S}}>1} (\mathfrak{n}^D_{j,\hat{S}}-1) \mathbbm{1}_{\mathfrak{n}^D_{j,\hat{S}}>1}}{{c \choose x}}- \left(\frac{n\sum\limits_{t=1}^{r-1} t{r \choose (t+1)} {(c-r) \choose (x-t-1)}}{{c \choose x}}\right)^2\nonumber\\
    & \overset{(c)}{=} \frac{\sum\limits_{i=1}^{n} \sum\limits_{j=1}^{n} \sum\limits_{\substack{\hat{S} \subset \mathcal{S};\\|\hat{S}|=x}} (\mathfrak{n}^D_{i,\hat{S}}-1) \mathbbm{1}_{\mathfrak{n}^D_{i,\hat{S}}>1} (\mathfrak{n}^D_{j,\hat{S}}-1) \mathbbm{1}_{\mathfrak{n}^D_{j,\hat{S}}>1}}{{c \choose x}}- \left(\frac{n{{c-r \choose x}}}{{{c \choose x}}} + n \times\left(\frac{rx}{c}-1\right)\right)^2
\end{align}

In the above set of equations, $(a)$ follows from the identity $(\sum\limits_i b_i)^2 = \sum\limits_i \sum\limits_j b_i b_j$. The first term in $(b)$ is obtained by interchanging the order of summations, whereas the second term comes from equation \eqref{removal_comp_mean}.  Further, the second term in $(c)$ follows using the equation \eqref{temp_result_removal} given in the proof of Theorem \ref{all_mean}.

Observe that the second term in the final expression in equation~\eqref{sigma_expression1} depends only on $n,r,c$ and $x$ and is independent of the specific balanced $(n,k,r,c)$ job assignment $D$. On the other hand, the first term in equation~\eqref{sigma_expression1} depends on the particular assignment $D$. We now consider the numerator of the first term of equation~\eqref{sigma_expression1} in more detail. We can break this expression into two parts, where one part is dependent on just one index $i$ and the other part is dependent on two distinct indices $i,j$. Thus,   \begin{align}{\label{twoparts}}
	\sum\limits_{i=1}^{n} \sum\limits_{j=1}^{n} \sum\limits_{\substack{\hat{S} \subset \mathcal{S};\\|\hat{S}|=x}} (\mathfrak{n}^D_{i,\hat{S}}-1) \mathbbm{1}_{\mathfrak{n}^D_{i,\hat{S}}>1} (\mathfrak{n}^D_{j,\hat{S}}-1) \mathbbm{1}_{\mathfrak{n}^D_{j,\hat{S}}>1}
	= & 2\sum\limits_{1\leq i < j\leq n} \sum\limits_{\substack{\hat{S} \subset \mathcal{S};\\|\hat{S}|=x}} (\mathfrak{n}^D_{i,\hat{S}}-1) \mathbbm{1}_{\mathfrak{n}^D_{i,\hat{S}}>1} (\mathfrak{n}^D_{j,\hat{S}}-1) \mathbbm{1}_{\mathfrak{n}^D_{j,\hat{S}}>1}\\
	& + \sum\limits_{1\leq i =j\leq n} \sum\limits_{\substack{\hat{S} \subset \mathcal{S};\\|\hat{S}|=x}} (\mathfrak{n}^D_{i,\hat{S}}-1) \mathbbm{1}_{\mathfrak{n}^D_{i,\hat{S}}>1} (\mathfrak{n}^D_{j,\hat{S}}-1) \mathbbm{1}_{\mathfrak{n}^D_{j,\hat{S}}>1}
\end{align}
In equation~\eqref{twoparts}, the second term can be rewritten as $\sum\limits_{1\leq i \leq n} \sum\limits_{\substack{\hat{S} \subset \mathcal{S};\\|\hat{S}|=x}} ((\mathfrak{n}^D_{i,\hat{S}}-1) \mathbbm{1}_{\mathfrak{n}^D_{i,\hat{S}}>1})^2$. For every job $a_i$, this expression calculates $\sum\limits_{\substack{\hat{S} \subset \mathcal{S};\\|\hat{S}|=x}} ((\mathfrak{n}^D_{i,\hat{S}}-1) \mathbbm{1}_{\mathfrak{n}^D_{i,\hat{S}}>1})^2$, which is independent of the choice of the job $a_i$ in any balanced $(n,k,r,c)$ assignment $D$. In fact, this second term of equation~\eqref{twoparts} is independent of the choice of $D$ and it depends only on the values of $n,c,r$ and $x$. We can compute this sum by counting the number of subsets $\hat{S} \subset \mathcal{S}$ of cardinality $x$ that additionally satisfy the constraint $\mathfrak{n}^D_{i,\hat{S}}=t$ (i.e. job $a_i$ is present in exactly $t$ servers from $\hat{S}$) for every $t$ from $2$ to $r$. Thus \begin{equation}\label{squares}
\sum\limits_{1\leq i \leq n} \sum\limits_{\substack{\hat{S} \subset \mathcal{S};\\|\hat{S}|=x}} ((\mathfrak{n}^D_{i,\hat{S}}-1) \mathbbm{1}_{\mathfrak{n}^D_{i,\hat{S}}>1})^2 = n\sum\limits_{t=1}^{r-1} t^2{r \choose (t+1)} {(c-r) \choose (x-t-1)}
\end{equation} 
Note that the number of subsets $\hat{S} \subset \mathcal{S}$ of cardinality $x$ such that a particular job $a_i$ appears $t+1$ times in $\hat{S}$ is given by ${r \choose (t+1)}{(c-r) \choose (x-t-1)}$. As $\mathfrak{n}^D_{i,\hat{S}}=t+1$ for this particular $\hat{S}$, therefore $((\mathfrak{n}^D_{i,\hat{S}}-1) \mathbbm{1}_{\mathfrak{n}^D_{i,\hat{S}}>1})^2 = t^2$. This explains the final expression in equation~\eqref{squares}.  A closed form expression for the sum $\sum\limits_{t=1}^{r-1} t^2{r \choose (t+1)} {(c-r) \choose (x-t-1)}$ can be obtained by considering the following binomial expressions 
\begin{equation}{\label{binomial_corollary_4}}
    r(r-1) y^2(1+y)^{r-2} - ry(1+y)^{r-1} - 1 + {(1+y)^r} = \sum_{t=1}^{r-1} t^2 {{r \choose t+1}} y^{t+1}
\end{equation}

\begin{equation}{\label{binomial_corollary_5}}
        (1+y)^{c-r} = \sum_{v=0}^{c-r} {{c-r \choose v}} y^v
\end{equation}

Multiplying equations \eqref{binomial_corollary_4} and \eqref{binomial_corollary_5}, one obtains $\sum\limits_{t=1}^{r-1} t^2 {{r \choose t+1}}{{c-r \choose x-t-1}}$ to be the coefficient of $y^{x}$ in $r(r-1) y^2(1+y)^{c-2} - ry(1+y)^{c-1} - (1+y)^{c-r} + (1+y)^c$, thus,  

\begin{equation}{\label{temp_corollary}}
        \sum\limits_{t=1}^{r-1} t^2 {{r \choose t+1}}{{c-r \choose x-t-1}} = r(r-1) {{c-2 \choose x-2}} - r{{c-1 \choose x-1}} -{{c-r \choose x}}+ {{c \choose x}}
\end{equation}

Thus, combining equations \eqref{temp_corollary}, \eqref{twoparts} and \eqref{sigma_expression1}, we get the expression in equation \eqref{sigma_expression_final}. Further, using $\mathfrak{m}^D(.)$ that denotes the frequency distribution of job pairs, we can say 
 \begin{equation}{\label{var_sim_exp}}
\sigma_{{D},x}(d)= \frac{2\sum\limits_{m=0}^{r}\mathfrak{m}^D(m)g(m,x) + T_2(n,r,c,x)}{{c \choose x}}- \left(T_1(n,r,c,x)\right)^2\\
\end{equation}  
 
{where} $T_1(n,r,c,x) = \frac{n{{c-r \choose x}}}{{{c \choose x}}} + n \times(\frac{rx}{c}-1)$ and 
\begin{align}
    T_2(n,r,c,x) = & n\sum\limits_{t=1}^{r-1} t^2{r \choose (t+1)} {(c-r) \choose (x-t-1)}\\
    = & n\left(r(r-1) {{c-2 \choose x-2}} - r{{c-1 \choose x-1}} -{{c-r \choose x}}+ {{c \choose x}}\right)
\end{align}

%% file: lemma_proximal_minimal.tex
\subsection{Proof of Lemma \ref{lemma_on_l_proximally_minimal}}\label{sec:proof_lemma_proximal_minimal}

\begin{lemma*}
    For proximally compact $(n,k,r,c)$ assignment $D$, we have $\ell=\Bigl\lfloor\frac{r(k-1)}{n-1}\Bigr\rfloor$ where $\ell$ denotes the index of the smallest non-zero entry in the shape vector $h_D$.
\end{lemma*}


\begin{proof}

For a proximally compact $(n,k,r,c)$ assignment $D$, if the shape vector has only one nonzero entry, then $\mathfrak{m}^D(\ell) = {n \choose 2}$. As the total number of job pairs at the $c$ servers is $c{k \choose 2}$, therefore using \eqref{property_m_2}, we have $\ell = \frac{c{k \choose 2}}{{n \choose 2}} = \frac{r(k-1)}{n-1}$. 

On the other hand, if a proximally compact $(n,k,r,c)$ assignment $D$, has a shape vector with exactly two consecutive nonzero entries, then $\mathfrak{m}^D(m)$ is zero for all $m\neq \ell,\ell+1$ for some $\ell$. Using \eqref{proporty_m_1}, we have $\mathfrak{m}^D(\ell+1) = {{n \choose 2}}- \mathfrak{m}^D(\ell)$ and by \eqref{property_m_2}, we get  $\ell(\mathfrak{m}^D(\ell)) + (\ell+1)({{n \choose 2}}- \mathfrak{m}^D(\ell)) = c{{k\choose 2}}$. Thus we can conclude that  $\ell {{n \choose 2}} \leq c{{k \choose 2}}$ and $(\ell +1){{n \choose 2}} > c{{k \choose 2}}$ (as $\mathfrak{m}^D(\ell)>0$) and therefore $\ell=\Bigl\lfloor\frac{c.{{k \choose 2}}}{{{n \choose 2}}}\Bigr\rfloor = \Bigl\lfloor\frac{r(k-1)}{n-1}\Bigr\rfloor$.
\end{proof}

%% file: min_var_proof.tex
\subsection{Proof of Theorem \ref{min_var}} {\label{sec:min_var_proof}}
\begin{theorem*}
 If a proximally compact balanced $(n,k,r,c)$ job assignment exists, then it has the least variance amongst all balanced $(n,k,r,c)$ job assignments. 
\end{theorem*}
\begin{proof}
Let $h_D$ be the shape vector corresponding to the proximally compact balanced $(n,k,r,c)$ job assignment. Thus $h_D(i) = 0$ for all $i \le \ell$ and $i > \ell + 2$ for $\ell$ as calculated in Lemma~\ref{lemma_on_l_proximally_minimal}.   Any balanced $(n,k,r,c)$ job assignment $D_1$ would have a shape vector $h_D + v$ where $v \in \ker H$ with matrix $H$ as defined in \eqref{defn_of_A}. Observe from the expression of variance in \eqref{var_sim_exp} that it is only the term $\sum\limits_{m=0}^{r}\mathfrak{m}^D(m)g(m,x)$ that varies amongst the different balanced $(n,k,r,c)$ assignments. Thus it is enough to show that for every permissible $v \in \ker H$ mentioned above,  $\sum\limits_{m=0}^{r}v(m+1)g(m,x) \ge 0$, in order to conclude that the proximally compact balanced $(n,k,r,c)$ assignment has the least variance. A basis for the kernel of the matrix $H$ is given by the $r-1$ vectors \begin{equation}{\label{basis_ker}} \left\{h_1, h_2, h_3, \cdots, h_{r-2}, h_{r-1} \right\} =  \left\{ \begin{pmatrix} 1 \\ -2 \\ 1 \\ 0 \\ \vdots \\ 0 \\ 0 \\ 0 \end{pmatrix},  \begin{pmatrix} 0 \\ 1 \\ -2 \\ 1 \\  \vdots \\ 0 \\ 0 \\ 0 \end{pmatrix}, \begin{pmatrix} 0 \\ 0 \\ 1 \\ -2 \\ \vdots \\ 0 \\ 0 \\ 0 \end{pmatrix}, \cdots , \begin{pmatrix} 0 \\ 0 \\ 0 \\ 0 \\ \vdots \\ -2 \\ 1 \\ 0 \end{pmatrix}, \begin{pmatrix} 0 \\ 0 \\ 0 \\ 0 \\ \vdots \\ 1 \\ -2 \\ 1 \end{pmatrix} \right\}\end{equation} 

We therefore first characterize $v \in \ker H$ that may appear from some balanced $(n,k,r,c)$ assignment. As both the shape vectors $h_D, h_D+v \in \mathbb{N}^{r+1}$, therefore $v(i) \ge 0$ for all $i \neq \ell+1, \ell + 2$. Further, as $v \in \ker H$, therefore $\sum\limits_{i=1}^{r+1} v(i) = 0$ and so if $v$ is a nonzero vector, then at least one of $v(\ell+1), v(\ell+2)$ must be a negative integer. As $v \in \ker H$, therefore $v$ can be expressed in terms of the basis vectors $\left\{h_i\right\}$ listed in \eqref{basis_ker}. Let $v = \sum\limits_{i=1}^{r-1}\alpha_ih_i$. We now show that all $\alpha_i \in \mathbb{N}$. 

Consider the components of $v$ for $i \le \ell$. Let $j$ be the smallest index where $v(j) > 0$ and $j \le \ell$. Then one can inductively argue that $\alpha_i = 0$ for all $i < j$ by starting with $i=1$ as $v(i) = 0$ for $i < j$. Further, $\alpha_j = v(j) > 0$. Now, $v(j+1) = \alpha_{j+1} -2\alpha_j \ge 0$ implies that $\alpha_{j+1} \ge 2\alpha_j$. Similarly, $v(j+2) = \alpha_{j+2} - 2\alpha_{j+1} + \alpha_j$, which implies $\alpha_{j+2} \ge 2\alpha_{j+1} - \alpha_{j} \ge 3\alpha_j$ as $v(j+2) \ge 0$. On the same lines, $\alpha_{j+3} \ge 2\alpha_{j+2} - \alpha_{j+1} \ge 2(2\alpha_{j+1} - \alpha_j) - \alpha_{j+1} = 3\alpha_{j+1} - 2\alpha_j \ge 4\alpha_j$. Inductively, one can show that $\alpha_{j+k} \ge (k+1)\alpha_j$ for all $k \le (\ell - j)$. Thus $\alpha_i \ge 0$ for all $1 \le i \le \ell$. This accounts for all $v(i) \ge 0$ for $i \le \ell$. 

Similarly, one can utilize $v(i) \ge 0$ for $i > \ell+2$ to conclude that $\alpha_i \ge 0$ for $\ell + 1 \le i \le r-1$, by proceeding from the other end. Let $j$ now be the largest index where $v(j) > 0$ and $j > \ell+2$. If $j < r+1$, then $v(r+1) = 0$ forces $\alpha_{r-1} = 0$. Once again, one can inductively argue that $\alpha_i = 0$ for $j-1 \le i \le r-1$ as the corresponding $v(i+2) = 0$. Further, $\alpha_{j-2} = v(j) > 0$. Now, $v(j-1) = \alpha_{j-3} - 2\alpha_{j-2} \ge 0$ implies that $\alpha_{j-3} \ge 2\alpha_{j-2}$. Using $v(j-2) \ge 0$, one obtains $\alpha_{j-4} \ge 2\alpha_{j-3} - \alpha_{j-2} \ge 3\alpha_{j-2}$. Reflecting the argument used before, one can conclude that $\alpha_{j-2-k} \ge (k+1)\alpha_{j-2}$ for $0 \le k \le (j - \ell - 3)$ and so $\alpha_i \ge 0$ for $\ell + 1 \le i \le r-1$. Thus all $\alpha_i \in \mathbb{N}$.


As a result, we get \begin{eqnarray*}
    \sum\limits_{m=0}^{r} v(m+1)g(m,x) = \sum\limits_{i=1}^{r-1} \alpha_i \left( g(i-1,x) -2g(i,x) + g(i+1,x) \right)
\end{eqnarray*}
By Lemma~\ref{g_p_diff_lemma}, we know that $g(m+1,x) - g(m,x) = {{c-2} \choose {x-1}} - 2{{c-r-1} \choose {x-1}} + {{c-2r+m} \choose {x-1}}$ and therefore \begin{eqnarray}
    g(i-1,x) -2g(i,x) + g(i+1,x) & = & g(i+1,x) - g(i,x) - \left( g(i,x) - g(i-1,x) \right) \nonumber \\
    & = & {{c-2r+i} \choose {x-1}} -  {{c-2r+i-1} \choose {x-1}} \nonumber \\ & = & \frac{x-1}{c-2r+i-x+1}{{c-2r+i-1} \choose {x-1}} \ge 0
\end{eqnarray}
Thus we have \begin{eqnarray}{\label{ineq_gre}}
    \sum\limits_{m=0}^{r} v(m+1)g(m,x) & = & \sum\limits_{i=1}^{r-1} \alpha_i \left( g(i-1,x) -2g(i,x) + g(i+1,x) \right) \nonumber \\ & = & \sum\limits_{i=1}^{r-1} \frac{\alpha_i(x-1)}{c-2r+i-x+1}{{c-2r+i-1} \choose {x-1}} \ge 0
\end{eqnarray}
As the above is true for every permissible $v \in \ker H$ such that $h_D + v$ is a shape vector, therefore we conclude that proximally compact balanced $(n,k,r,c)$ assignment has the least variance amongst all balanced $(n,k,r,c)$ assignments.
\end{proof}

%% file: max_variance_proof.tex
\subsection{Proof of Theorem \ref{max_var}}\label{sec:max_var_proof}
\begin{theorem*}
 If a stretched compact balanced $(n,k,r,c)$ job assignment exists, then it has the largest variance amongst all balanced $(n,k,r,c)$ job assignments. 
\end{theorem*}
\begin{proof}
   Let $h_D$ be the shape vector corresponding to a stretched compact balanced $(n,k,r,c)$ job assignment. Then $h_D(i) = 0$ for all $i \neq 1, r+1$. Any balanced $(n,k,r,c)$ job assignment $D_1$ has a shape vector $h_D + v$ where $v \in \ker H$ with $H$ as defined in \eqref{defn_of_A}. Similar to the proof of Theorem~\ref{min_var}, we now characterize $v \in \ker H$ that can arise from some balanced $(n,k,r,c)$ assignment. Note that $v(i) \ge 0$ for all $i \neq 1,r+1$ and $\sum\limits_{i=1}^{r+1} v(i) = 0$. Hence $v(1) \le 0$ and $v(r+1) \le 0$. As $v \in \ker H$, therefore $v = \sum\limits_{i=1}^{r-1}\alpha_ih_i$ where $\left\{h_i\right\}$ is the basis for $\ker H$ as listed in \eqref{basis_ker}. Following the proof of Theorem~\ref{min_var}, it is now enough to show that all $\alpha_i \le 0$ for any permissible nonzero vector $v \in \ker H$ since the expression in \eqref{ineq_gre} would then be rendered negative, thereby signalling a decrease in the variance for all balanced $(n,k,r,c)$ assignments with shape vector $h_D + v$. 

   Observe that as $v(1) \le 0$, therefore $\alpha_1 = v(1) \le 0$. If $r = 2$, then $\ker H$ is one dimensional and therefore $v(1) = v(3) \le 0$, while $v(2) = -2v(1) \ge 0 \ge v(1)$. If $r > 2$, then as $v(2) \ge 0$ and  $v(2) = \alpha_2 - 2\alpha_1$, therefore $\alpha_2 \ge 2\alpha_1$. Now if $r = 3$, then $\ker H$ is two dimensional and $\alpha_2 = v(4) \le 0$. Further, $v(3) \ge 0$ implies $\alpha_1 - 2\alpha_2 \ge 0$ and therefore $\alpha_1 \ge 2\alpha_2$. Thus again, for $r = 3$, both $\alpha_1, \alpha_2$ are negative and their values are mutually bound by the constraints $\frac{\alpha_1}{2} \ge \alpha_2 \ge 2\alpha_1$ and $\frac{\alpha_2}{2} \ge \alpha_1 \ge 2\alpha_2$.

   Now, we consider the cases of $r > 3$. In this case, $v(3) \ge 0$ translates to $2\alpha_2 \le \alpha_1 + \alpha_3$. Therefore $4\alpha_2 \le 2\alpha_1 + 2\alpha_3 \le \alpha_2 + 2\alpha_3$ (by using the condition obtained from $v(2) \ge 0$) which in turn translates to $3\alpha_2 \le 2\alpha_3$. From the condition $v(4) \ge 0$, we get $2\alpha_3 \le \alpha_2 + \alpha_4$ which can now be manipulated to $6\alpha_3 \le 3\alpha_2 + 3\alpha_4 \le 2\alpha_3 + 3\alpha_4$ which gives us $4\alpha_3 \le 3\alpha_4$. Following the steps mentioned above, one can use the subsequent $v(k) \ge 0$ to show that $k\alpha_{k-1} \le (k-1)\alpha_k$ for $2 \le k \le (r-1)$. Combining all these inequalities, one gets $\alpha_1 \le \frac{\alpha_2}{2} \le \frac{\alpha_3}{3} \le \cdots \le \frac{\alpha_k}{k} \le \cdots \le \frac{\alpha_{r-1}}{r-1} \le 0$. This proves that all the $\alpha_i \le 0$. 
\end{proof}

Interestingly, in the proof above, one could have started from the other end and as already shown for the case $r = 3$, one can get another set of constraints $\alpha_{r-1} \le \frac{\alpha_{r-2}}{2} \le \cdots \le \frac{\alpha_{r-k}}{k} \le \cdots \le \frac{\alpha_1}{r-1} \le 0$. These interwoven constraints restrict the possible values for the $\alpha_i$ where $v = \sum\limits_{i=1}^{r-1}\alpha_ih_i$. 

%% file: g0_expr_proof.tex
\subsection{Proof of Lemma \ref{g_p_0_lemma}}\label{sec:g_p_0_lemma_proof}

\begin{lemma*}
    For a balanced $(n,k,r,c)$ assignment, the value of $g(0,x)$ is given by 
    \begin{align}{\label{g_0_exp}}
    g(0,x) =  r^2 {{c-2 \choose x-2}} & - 2r {{c-1 \choose x-1}}  + {{c \choose x}} - 2{{c-r \choose x}} + 2 r{{c-r-1 \choose x-1}} + {c-2r \choose x}
\end{align}
\end{lemma*}

\begin{proof}

 Consider a pair of jobs $(a_i,a_j)$ such that no server has been assigned both $a_i$ and $a_j$ together. Therefore there are precisely $r$ servers that have been assigned $a_i$ and not $a_j$. Another $r$ servers that are assigned $a_j$ but not $a_i$ while the remaining $c-2r$ servers are assigned neither $a_i$ nor $a_j$. Then \begin{equation}\label{gfirst}
   g(0,x) = \sum\limits_{\substack{\hat{S} \subset \mathcal{S};\\|\hat{S}|=x}} (\mathfrak{n}^D_{i,\hat{S}}-1) \mathbbm{1}_{\mathfrak{n}^D_{i,\hat{S}}>1} (\mathfrak{n}^D_{j,\hat{S}}-1) \mathbbm{1}_{\mathfrak{n}^D_{j,\hat{S}}>1} = \sum_{t=2}^r \sum_{u=2}^r(t-1)(u-1){r \choose t}{r \choose u}{c-2r \choose x-t-u} 
\end{equation}
Clearly any subset of servers $\hat{S} \subset \mathcal{S}$ of cardinality $x$ that has at most only one instance of the job $a_i$ assigned amongst its members does not contribute to the sum. Ditto for $a_j$. Therefore, one needs to consider only those subsets $\hat{S}$ of servers that contain at least two servers that are assigned $a_i$ and at least two servers that are assigned $a_j$. In the final expression of equation~\eqref{gfirst}, ${r \choose t}{r \choose u}{c-2r \choose x-t-u}$ counts the number of subsets of servers $\hat{S}$ of cardinality $x$ that contain $t$ servers assigned $a_i$, $u$ servers assigned $a_j$ and $x-t-u$ servers that have been assigned neither. The summation limits ensure that there are at least $2$ servers assigned $a_i$ and at least $2$ servers assigned $a_j$. The expression $(t-1)(u-1)$ is the contribution of each subset $\hat{S}$ that contains $t$ copies of $a_i$ and $u$ copies of $a_j$ assigned to its members. A closed form solution of the expression for $g(0,x)$ can be obtained by considering 
\begin{equation}{\label{binomial_corollary_1}}
        ry(1+y)^{r-1} + 1 - (1+y)^r = \sum_{t=1}^{r} (t-1) {{r \choose t}} y^t
\end{equation}

\begin{equation}{\label{binomial_corollary_2}}
        ry(1+y)^{r-1} + 1 - (1+y)^r = \sum_{u=1}^{r} (u-1) {{r \choose u}} y^u
\end{equation}

\begin{equation}{\label{binomial_corollary_3}}
        (1+y)^{c-2r} = \sum_{v=0}^{c-2r} {{c-2r \choose v}} y^v
\end{equation}

Multiplying these three expressions \eqref{binomial_corollary_1}, \eqref{binomial_corollary_2} and \eqref{binomial_corollary_3}, we get $\sum_{t=2}^{r}\sum_{u=2}^{r} (t-1)(u-1) {{c-2r \choose x-t-u}}{{r \choose t}}{{r \choose u}}$ to be the coefficient of $y^x$ in $\left(ry(1+y)^{r-1} + 1 - (1+y)^r\right)^2(1+y)^{c-2r}$. 

\begin{align}
    & \sum_{t=2}^{r}\sum_{u=2}^{r} (t-1)(u-1) {{c-2r \choose x-t-u}}{{r \choose t}}{{r\choose u}}\nonumber\\ 
    = &r^2 {{c-2 \choose x-2}} - 2r {{c-1 \choose x-1}} + {{c \choose x}} - 2{{c-r \choose x}} + 2 r{{c-r-1 \choose x-1}} + {c-2r \choose x}
\end{align}

\end{proof}

%% file: g_p_diff_proof.tex
\subsection{Proof of Lemma \ref{g_p_diff_lemma}}\label{sec:g_p_diff_lemma_proof}
\begin{lemma*}
For a balanced $(n,k,r,c)$ assignment, the values for $g(m,x)$ are related in the following fashion 
\begin{align}{\label{g_p_diff}}
    g(m+1,x) - g(m,x) = {c-2 \choose x-1} & - 2{c-r-1 \choose x-1} + {c-2r+m \choose x-1}
\end{align}
\end{lemma*}
\begin{proof}
 Let us consider a pair of jobs $(a_i,a_j)$ that have been assigned together to precisely $m$ servers. Without loss of generality, let $s_1, s_2, \cdots, s_m$ be the servers that are assigned both the jobs $a_i, a_j$. Let $s_{m+1}, s_{m+2}, \cdots, s_r$ be the servers that have been assigned $a_i$ but not $a_j$. Assume servers $s_{r+1}, s_{r+2}, \cdots, s_{2r-m}$ are the servers assigned $a_j$ but not $a_i$. The last $c-2r+m$ servers $s_{2r-m+1}, s_{2r-m+2}, \cdots, s_c$ are the ones that have not been assigned $a_i$ or $a_j$. 
 
 Let another pair of jobs $(a_{i_1},a_{j_1})$ be such that they have been assigned together to precisely $m+1$ servers. We now consider a bijective map $f : \mathcal{S} \rightarrow \mathcal{S}$ described in the following fashion. Let $f(s_{\ell})$ for $1 \le \ell \le m+1$ be servers that have been assigned both the jobs $a_{i_1}$ and $a_{j_1}$. Let $f(s_{\ell})$ for $m+2 \le \ell \le r$ be servers that have been assigned $a_{i_1}$ but not $a_{j_1}$. Further let $f(s_{\ell})$ for $r+2 \le \ell \le 2r-m$ be servers that have been assigned $a_{j_1}$ but not $a_{i_1}$. The rest of $f(s_{\ell})$ have not been assigned $a_{i_1}$ or $a_{j_1}$. Thus there are two special servers, namely $s_{m+1}$ (which does job $a_i$ but not $a_j$) and $s_{r+1}$ (which does job $a_j$ but not $a_i$), and whose images $f(s_{m+1})$ (which does both the jobs $a_{i_1}$ and $a_{j_1}$) and $f(s_{r+1})$ (which does neither $a_{i_1}$ nor $a_{j_1}$) that we shall pay special attention to.

 For any $\hat{S} \subset \mathcal{S}$ of cardinality $x$, let us compare its contribution to the sum $g(m,x)$ with the contribution of $f(\hat{S})$ towards $g(m+1,x)$. Clearly, if $\hat{S} \subset \mathcal{S}\setminus \left\{ s_{m+1}, s_{r+1} \right\}$, then the contribution of $\hat{S}$ towards $g(m,x)$ is exactly the same as the contribution of $f(\hat{S})$ to $g(m+1,x)$. Similarly, if $s_{m+1}, s_{r+1} \in \hat{S}$, then contribution of $\hat{S}$ towards $g(m,x)$ and that of $f(\hat{S})$ towards $g(m+1,x)$ is exactly the same. Therefore it suffices to only consider those subsets $\hat{S}$ of cardinality $x$ that contain exactly one of the two special servers $\left\{ s_{m+1}, s_{r+1} \right\}$ to evaluate the difference $g(m+1,x) - g(m,x)$. Hence we look at subsets $\hat{S}$ that are formed by taking either $s_{m+1}$ or $s_{r+1}$ along with  $\bar{S} \subset \mathcal{S} \setminus \left\{ s_{m+1},s_{r+1} \right\}$ of cardinality $x-1$.  

 Let $\bar{S} \subset \mathcal{S} \setminus \left\{ s_{m+1},s_{r+1} \right\}$ of cardinality $x-1$ contain $\alpha > 0$ instances of job $a_i$ and $\beta > 0$ instances of job $a_j$ assigned to its servers. Then $\bar{S} \cup \{ s_{m+1} \}$ contributes $\alpha(\beta - 1)$ towards $g(m,x)$, whereas $\bar{S} \cup \{ s_{r+1} \}$ contributes $(\alpha - 1)\beta$ towards $g(m,x)$. At the same time, $f(\bar{S}) \cup \{ f(s_{m+1}) \}$ contributes $\alpha\beta$ towards $g(m+1,x)$, whereas $f(\bar{S}) \cup \{ f(s_{r+1}) \}$ contributes $(\alpha - 1)(\beta - 1)$ towards $g(m+1,x)$. Thus, one can evaluate the contribution of $\bar{S}$ towards the difference $g(m+1,x)-g(m,x)$ to be $\alpha\beta + (\alpha - 1)(\beta - 1) - \alpha(\beta - 1) - (\alpha - 1)\beta = 1$. So every subset $\bar{S} \subset \mathcal{S} \setminus \{ s_{m+1}, s_{r+1} \}$ of cardinality $x-1$, whose servers have at least one instance each of jobs $a_i$ and $a_j$ assigned to them, contributes a net change of $1$ towards the difference $g(m+1,x)-g(m,x)$. One needs to just count the number of subsets $\bar{S}$ of cardinality $x-1$ that satisfy these conditions to find $g(m+1,x) - g(m,x)$. 

 Total number of subsets of cardinality $x-1$ of the set $\mathcal{S} \setminus \{ s_{m+1}, s_{r+1} \}$ is given by ${c-2 \choose x-1}$. If the subset $\bar{S}$ is one of the ${c-r-1 \choose x-1}$ subsets chosen from the servers $\{s_{r+2}, s_{r+3}, \cdots s_c \}$, then the job $a_i$ is not assigned to any of its servers. Similarly, if $\bar{S}$ is one of the ${c-r-1 \choose x-1}$ chosen from the servers $\{s_{m+2}, s_{m+3}, \cdots, s_{r} \} \cup \{ s_{2r-m+1}, s_{2r-m+2}, \cdots, s_c\}$, then it does not have any instance of the job $a_j$ assigned to its servers. As these subsets $\bar{S}$ do not contribute to the difference $g(m+1,x)-g(m,x)$, their numbers have to be subtracted from ${c-2 \choose x-1}$. In the process, $\bar{S} \subset \{ s_{2r-m+1}, s_{2r-m+2}, \cdots, s_c \}$ have been subtracted twice and therefore  ${c-2r+m \choose x-1}$ needs to added back (inclusion-exclusion principle), thereby giving \begin{equation*}
     g(m+1,x) - g(m,x) = {c-2 \choose x-1} - 2{c-r-1 \choose x-1} + {c-2r+m \choose x-1}
 \end{equation*} 
\end{proof}

%% file: max_min_var_sampled.tex
\subsection{Proof of Theorem \ref{all_mean_sampled} and Theorem \ref{max_min_var_sampled}} \label{sec:max_min_var_sampled_proof}

\begin{theorem*}
Consider any balanced $(n,k,r,c)$ assignment $D$, where each server is independently and equally likely to straggle with probability $p$. The expectation of the number of distinct completed jobs $d$ received is 
 
 \begin{equation}
  \mathbbm{E}_{D}[d]= n- n(p)^r
 \end{equation}
 \end{theorem*}

 \begin{proof}
     Observe that under this setup, the number of servers that communicates with the master $x$ is be given by the  binomial distribution $B(c,1-p)$. Also, observe that under this setup conditioned on $x$, any set of $x$ servers is equally likely to communicate with the master.

     We can thus say that 
\begin{align}
    \mathbbm{E}_{D}[d]= & \mathbbm{E}_{x\sim B(c,p)}\mathbbm{E}_{D,x}[d] \\
    & \overset{(a)}{=} \sum\limits_{x=0}^{c} n\left(1-\frac{{c-r \choose x}}{{c \choose x}}\right) {{c \choose x}} p^{c-x}(1-p)^{x}\\
    & = n - n(1-p)^r \sum\limits_{x=0}^{c} {{c-r \choose x}} p^{(c-r-x)}(1-p)^{x} = n- n(1-p)^r
\end{align}

Note $(a)$ follows from the expression of mean in Theorem \ref{all_mean}. 
 \end{proof}

\begin{theorem*}
    Let us consider $x \sim \mathcal{P}$. Conditioned on $x$, we study the setup where any set of $x$ servers is equally likely to communicate with the master. Then the proximally compact assignment (if it exists) and stretched compace assignment (if exists) attain the least and the largest variance respectively on the number of distinct jobs received at master amongst all balanced $(n,k,r,c)$ assignment schemes.
\end{theorem*}

\begin{proof}
    Let us denote the number of distinct jobs when any set of $x$ servers return uniformly at random by $d$. However, in our problem $x$ itself might be sampled from a distribution $\mathcal{P}$. Let us denote the variance in this set-up under this assignment of jobs to servers (say ${D}$) by $\sigma_{{D},x \sim \mathcal{P}}(d)$.

    Now using law of variances(Eve's law), we can say that $$ \sigma_{{D},x \sim \mathcal{P}}(d) = \mathbb{E}_{x \sim \mathcal{P}}[\sigma_{{D},x}(d)] + \sigma_{x \sim \mathcal{P}}[\mathbb{E}_{{D},x}(d)] $$ Now consider assignments $D$ and ${D}_1$ such that assignment $D$ is a proximally compact $(n,k,r,c)$ assignment scheme and assignment $D_1$ could be any balanced $(n,k,r,c)$ assignment scheme.

    However, we know from Theorem \ref{min_var} that $\sigma_{{D},x}(d) \leq \sigma_{{D}_1,x}(d)$
     for every $x$ if $D$ is a proximally compact $(n,k,r,c)$ assignment scheme and assignment $D_1$ is any other balanced $(n,k,r,c)$ assignment scheme.

     We also know that $\mathbbm{E}_{{D},x}(d) = \mathbbm{E}_{{D}_1,x}(d)$ from Theorem \ref{all_mean}. Combining the two properties, we get that $\sigma_{{D},x \sim \mathcal{P}}(d) \leq \sigma_{{D}_1,x \sim \mathcal{P}}(d)$, thus proving that proximally compact schemes (if exists) attain the least variance. A very similar approach can be used to prove the result on maximal variance as well.

\end{proof}

%% file: appendix-corollary.tex
\subsection{Proof of Corollary \ref{corollary_variance}}{\label{proof_corollary_variance}}

\begin{namedcorollary}{Corollary}{\label{corollary_variance}}
    For $x>c-r$ , the expression of $\sigma_{D,x}(d)$ in Equation \eqref{var_sim_exp} equals zero.
\end{namedcorollary}

\begin{proof}
    Recall the expression of $\sigma_{D,x}(d)$ from Equation \eqref{var_sim_exp}. Observe that expression $g(m+1,x)-g(m,x)$ from Equation \eqref{g_p_diff} would be ${{c-2 \choose x-1}}$ for $x>c-r$ as the second and third term in equation \eqref{g_p_diff} goes to zero since $p\leq r$ and $x>c-r$. 

\begin{equation}{\label{g_p_exp_claim}}
    g(m+1,x)-g(m,x) = {{c-2 \choose x-1}}
\end{equation}

Let us now compute $g(0,x)$ using the expression in \eqref{gfirst} and \eqref{g_0_exp} for $x>c-r$.

\begin{equation}{\label{g_0_exp_claim}}
    g(0,x)=\sum_{i=2}^{r+1}\sum_{j=2}^{r+1} (i-1)(j-1) {{c-2r \choose x-i-j}} = r^2 {{c-2 \choose x-2}} - 2r {{c-1 \choose x-1}} + {{c \choose x}}
\end{equation}

Thus, from equations \eqref{g_p_exp_claim} and \eqref{g_0_exp_claim}, we get 

\begin{equation}{\label{g_p_expr}}
    g(m,x) = r^2 {{c-2 \choose x-2}} - 2r {{c-1 \choose x-1}} + {{c \choose x}} + m \times {{c-2 \choose x-1}}
\end{equation}

Since, $x>c-r$, we may claim that the term $T_2(n,k,r,c)$ in Equation \eqref{var_sim_exp} goes as follows.

\begin{equation}{\label{T_2_exp_claim}}
    T_2(n,k,r,c) = \sum\limits_{t=1}^{r-1} t^2{r \choose (t+1)} {(c-r) \choose (x-t-1)} = \left(r(r-1) {{c-2 \choose x-2}} - r{{c-1 \choose x-1}} + {{c \choose x}}\right) 
\end{equation}

Also observe that since $x>c-r$ the term ${{c-r \choose x}}$ goes to zero, hence not written in equation \eqref{temp_corollary}. Thus the numerator of the first term in equation \eqref{var_sim_exp} is given by (from equations \eqref{g_0_exp_claim} and \eqref{g_p_expr} and \eqref{T_2_exp_claim})

\begin{align}{\label{temp1_corollary}}
& 2.\sum\limits_{m=0}^{r}\mathfrak{m}^{D}(m)g(m,x) + n\sum\limits_{t=1}^{r-1} t^2{r \choose (t+1)} {(c-r) \choose (x-t-1)}\nonumber\\
\hspace{-1 em}=  & \sum\limits_{m=0}^{r} \Biggl(2\mathfrak{m}^{D}(m)\left(r^2 {{c-2 \choose x-2}} - 2r {{c-1 \choose x-1}} + {{c \choose x}}\right) 
+ 2m \mathfrak{m}^{D}(m) {{c-2 \choose x-1}}\Biggr)\nonumber \\ &\hspace{ 18 em}+  n\left(r(r-1) {{c-2 \choose x-2}} - r{{c-1 \choose x-1}} + {{c \choose x}}\right) \nonumber\\
& \overset{(a)}{=} \left(r^2 {{c-2 \choose x-2}} - 2r {{c-1 \choose x-1}} + {{c \choose x}}\right) n(n-1) + {{c-2 \choose x-1}} ck(k-1) \nonumber\\
&  \hspace{ 18 em}+ n\left(r(r-1) {{c-2 \choose x-2}} - r{{c-1 \choose x-1}} + {{c \choose x}}\right)\nonumber\\
& \overset{(b)}{=} {{c-2 \choose x-2}} (nr(nr-1)) + {{c-2 \choose x-1}} ck(k-1) - {{c-1 \choose x-1}} (nr(2n-1)) + n^2{{c \choose x}}\nonumber\\
& \overset{(c)}{=} {{c-2 \choose x-2}} n^2r^2 + {{c-2 \choose x-1}} ck^2 - {{c-1 \choose x-1}} (nr(2n)) + n^2{{c \choose x}} \nonumber\\
& \overset{(d)}{=} {{c-1 \choose x-1}} {nr \times kx} - {{c-1 \choose x-1}} (nr(2n)) + n^2{{c \choose x}} \nonumber\\
& \overset{(e)}{=} {{c \choose x}} \left(\left(\frac{nrx}{c}\right)^2- 2n \left(\frac{nrx}{c}\right) + n^2  \right)\nonumber \\
& \overset{(f)}{=} {{c \choose x}} \left(n \times \left(\frac{rx}{c}-1\right)\right)^2
\end{align}

We now argue for each of the steps below.

\begin{itemize}
    \item $(a)$ follows since $\sum\limits_{m=0}^{r} m\times \mathfrak{m}^D(m) = c {{k \choose 2}}$ and $\sum\limits_{m=0}^{r} \mathfrak{m}^D(m) = {{n \choose 2}}$ in Equations \eqref{proporty_m_1} and \eqref{property_m_2}
    
    

    \item $(b)$ follows by combining the coefficints of ${{c-2 \choose x-2}}$, ${c-1 \choose x-1}$ and ${c \choose x}$.   

    \item $(c)$ follows as $nr{{c-2 \choose x-2}} + kc{{c-2 \choose x-1}} = nr {{c-1 \choose x-1}}$. This can be explained by the fact that $n \times r = k \times c$. 

    \item $(d)$ follows from the following set of equalities

    \begin{align*}
        {{c-2 \choose x-2}} n^2r^2 + {{c-2 \choose x-1}} ck^2 = & \frac{(c-1)!}{(x-2)!}{(c-x-1)!} \left(\frac{ck}{c-x}+\frac{k}{x-1}\right) \\ = & \frac{(c-1)!\times nr\times kx(c-1)}{(x-2)!(c-x)(x-1)} = nr \times kx {{c-1 \choose x-1}}
    \end{align*}

    \item $(e)$ and $(f)$ follow from the fact that $n \times r = k \times c$

\end{itemize}

Now, observe the second term of $\sigma_{D,x}(d)$ in equation \eqref{var_sim_exp} and we see that $T_1(n,k,r,c)=n\times \left(\frac{rx}{c}-1\right)$ as $x>c-r$. Thus, using equation \eqref{temp1_corollary}, we can say that $\sigma_{D,x}(d)=0$ for $x>c-r$.

\end{proof}


\remove{
We now state and prove the result in Theorem \ref{mean_equality}.
\begin{theorem*}
 Consider any assignment $D$ amongst balanced $(n,k,r,c)$ assignments. Note that the expectation of the number of distinct jobs $d$ received at the master when any subset of servers $\mathcal{S}$ of cardinality $x>1$ is able to communicate to the master with equal probability would be the same for every assignment $D$ amongst all balanced $(n,k,r,c)$ assignments. We state it as an equation below.

 \begin{equation}
  \mathbbm{E}_{{D},x}[d]=n\times \left(1-\frac{{c-r \choose x}}{{c \choose x}}\right) 
 \end{equation}
 
 Additionally, the variance of the number of distinct jobs can be written as: 
 \begin{equation}{\label{var_exp}}
 \sigma_{{D},x}(d) = \frac{\sum\limits_{i=1}^{n} \sum\limits_{j=1}^{n} \sum\limits_{\substack{\hat{S} \subset \mathcal{S};\\|\hat{S}|=x}} (\mathfrak{n}^D_{i,\hat{S}}-1) \mathbbm{1}_{\mathfrak{n}^D_{i,\hat{S}}>1} (\mathfrak{n}^D_{j,\hat{S}}-1) \mathbbm{1}_{\mathfrak{n}^D_{j,\hat{S}}>1}}{{{c \choose x}}}- \left(\frac{n\sum\limits_{t=1}^{r-1} t{r \choose (t+1)} {(c-r) \choose (x-t-1)}}{{c \choose x}}\right)^2
 \end{equation}\\

 where $\mathfrak{n}^D_{i,S}$ denotes the number of servers in set $S$ which have been assigned the job $a_i$ under the assignment scheme $D$.

 Also under the constraint $n=c$, we obtain $\mathbbm{E}_{{D},x}[d]=n.\left(1 - \frac{{{n-x \choose r }}}{{{n \choose r}}}\right)$
 		
\end{theorem*}

}
\remove{
	
	\begin{theorem}{\label{min_var_int_1}}
		For a given $n,k,c,r$, any distribution which satisfies the properties in \ref{setting_condition} with no two servers having more then one common job has the same variance of distinct jobs received for every $x \leq c$. 		
	\end{theorem}
	
	\begin{proof}
		We now consider the expression of variance from equation \eqref{var_exp} and compute $f(i,j)$ (defined below) for every unordered pair $(i,j)$ as the second term remains the same irrespective of distribution $D$.
		
		\begin{equation}{\label{f_defn}}
		f(i,j) = \sum\limits_{\substack{\hat{S} \subset \mathcal{S};\\|\hat{S}|=x}} (\mathfrak{n}^D_{i,\hat{S}}-1) \mathbbm{1}_{\mathfrak{n}^D_{i,\hat{S}}>1} (\mathfrak{n}^D_{j,\hat{S}}-1) \mathbbm{1}_{\mathfrak{n}^D_{j,\hat{S}}>1}
		\end{equation}
		
		\remove{

		$(a)$ follows since $\sigma(c-X) = \sigma(X)$ where $c$ is a constant and $X$ is the random variable. 
		The first term in $(b)$ follows since $(\sum\limits_i b_i)^2 = \sum\limits_i \sum\limits_j b_i b_j$. The first term in $(c)$ follows from interchange of summations whereas the second term follows since it is a constant irrespective of distribution $D$ as proven in proof of Theorem \ref{mean_equality}.

		Let us know compute $\sum\limits_{\substack{\hat{S} \subset \mathcal{S};\\|\hat{S}|=x}} (\mathfrak{n}^D_{i,\hat{S}}-1) \mathbbm{1}_{\mathfrak{n}^D_{i,\hat{S}}>1} (\mathfrak{n}^D_{j,\hat{S}}-1) \mathbbm{1}_{\mathfrak{n}^D_{j,\hat{S}}>1}$ for every unordered pair $(i,j)$.
	}

		\begin{enumerate}
			\item   Case 1: $i=j$
			
			In this case the expression becomes $\sum\limits_{\substack{\hat{S} \subset \mathcal{S};\\|\hat{S}|=x}} ((\mathfrak{n}^D_{i,\hat{S}}-1) \mathbbm{1}_{\mathfrak{n}^D_{i,\hat{S}}>1})^2$ which can exactly be computed as in the previous theorem \ref{mean_equality} and can be shown to be $\sum\limits_{t=1}^{r-1} t^2{r \choose (t+1)} {(c-r) \choose (x-t-1)}$
			
			\item Case 2: $i\neq j$
			
			 Now,
			 
			 \begin{align*}
			  \sum\limits_{\substack{\hat{S} \subset \mathcal{S};\\|\hat{S}|=x}} (\mathfrak{n}^D_{i,\hat{S}}-1) \mathbbm{1}_{\mathfrak{n}^D_{i,\hat{S}}>1} (\mathfrak{n}^D_{j,\hat{S}}-1) \mathbbm{1}_{\mathfrak{n}^D_{j,\hat{S}}>1}
			 \overset{(c)}{=} & \sum\limits_{u=2}^{r} \sum\limits_{v=2}^{r} \sum\limits_{\substack{\hat{S} \subset \mathcal{S};|\hat{S}|=x;\\\mathfrak{n}^D_{i,\hat{S}}=u;\mathfrak{n}^D_{j,\hat{S}}=v}} (u-1)\times (v-1)\\
			 \overset{(d)}{=} & \sum\limits_{u=2}^{r} \sum\limits_{v=2}^{r}  (u-1)\times (v-1) \sum\limits_{\substack{\hat{S} \subset \mathcal{S};|\hat{S}|=x;\\\mathfrak{n}^D_{i,\hat{S}}=u;\mathfrak{n}^D_{j,\hat{S}}=v}} 1
			 \end{align*}
			 
			 $(c)$ follows since both $\mathfrak{n}^D_{j,\hat{S}}$ and $\mathfrak{n}^D_{i,\hat{S}}$ is bounded above by $r$ whereas $(d)$ follows since $(u-1)(v-1)$ is a constant for the second summation.
			 
			 Let us know compute $\sum\limits_{\substack{\hat{S} \subset \mathcal{S};|\hat{S}|=x;\\\mathfrak{n}^D_{i,\hat{S}}=u;\mathfrak{n}^D_{j,\hat{S}}=v}} 1$ i.e, the number of subsets of servers with job $a_i$ occurring $u$ times and job $a_j$ occuring $v$ times for two different conditions on the pair $(i,j)$.
			 
			 \begin{itemize}
			 	\item Case (i): There exists no such server with both jobs $a_i$ and $a_j$.

			 	Now we know that there are exactly $r$ servers where job $a_i$ is present, similarly there are exactly $r$ servers where job $a_j$ is present.
			 	
			 	From the first $r$ servers, we have to choose exactly $u$ of them, similarly from the other $r$ servers we have to choose exactly $v$ of them and from the remaining $(c-2r)$ servers we have to choose exactly $(x-u-v)$ of them, thus there exist ${{r \choose u}}{{r \choose v}}{{(c-2r) \choose (x-u-v)}}$ subsets satisfying the desired property above.
			 	
			 	Thus the expression in $(d)$ turns out to be $\sum\limits_{u=2}^{r} \sum\limits_{v=2}^{r}  (u-1)\times (v-1) {{r \choose u}}{{r \choose v}}{{(c-2r) \choose (x-u-v)}}$ when the $(i,j)$ satisfies the condition in 2(i).

			 	\item Case (ii): There exists a server with both jobs $a_i$ and $a_j$ present. 
			 	
			 	Note that there can exist only one such server present as no two servers can have more than one job common.

			 	Now there exist $(r-1)$ servers with only job $a_i$ is present whereas there exist $(r-1)$ servers with only job $a_j$ present. Also there exist exactly one server with both jobs present. No other server contains either job $a_i$ or job $a_j$.

			 	Let us consider various ways these subsets of $x$ servers can occur.
			 	
			 	\begin{enumerate}
			 		\item The server with both jobs $a_i$ and $a_j$ is not present in the subset $\hat{S}$ at all.
			 		
			 		Thus from the first $(r-1)$ servers exactly $u$ of them must be present whereas from the other $(k-1)$ servers exactly $v$ of them must be present and from the $(c-2r+1)$ servers (containing neither $a_i$ nor $a_j$) only $(x-u-v)$ of them must be present, thus there exist exactly ${{(r-1} \choose u}{{(r-1) \choose v}}{{(c-2r+1) \choose (x-u-v)}}$ of such subsets.
			 		
			 		\item The server with both jobs $a_i$ and $a_j$ is present in the subset $\hat{S}$.
			 		
			 		Thus from the first $(r-1)$ servers exactly $(u-1)$ of them must be present, whereas from the other $(r-1)$ servers exactly $(v-1)$ of them must be present and from the $(c-2r+1)$ servers (containing neither $a_i$ nor $a_j$ ) only $(x-u-v+1)$ of them must be present, thus there exist exactly ${{(r-1} \choose (u-1)}{{(r-1) \choose (v-1)}}{{(c-2r+1) \choose (x-u-v+1)}}$ of such subsets.

			 	\end{enumerate}

			 	Thus there exist ${{(r-1} \choose u}{{(r-1) \choose v}}{{(c-2r+1) \choose (x-u-v)}} + {{(r-1} \choose (u-1)}{{(r-1) \choose (v-1)}}{{(c-2r+1) \choose (x-u-v+1)}}$ such subsets satisfying the desired property.
			 	
			 	Thus the expression in $(d)$ turns out to be $\sum\limits_{u=2}^{r} \sum\limits_{v=2}^{r}  (u-1)\times (v-1) \left({{(r-1} \choose u}{{(r-1) \choose v}}{{(c-2r+1) \choose (x-u-v)}} + {{(r-1} \choose (u-1)}{{(r-1) \choose (v-1)}}{{(c-2r+1) \choose (x-u-v+1)}}\right)$ when $(i,j)$ satisfies the condition in 2(ii). 
			 	
			 \end{itemize}
			 
		\end{enumerate} 
		
	It is clear that number of such pairs $(i,j)$ satisfying the case 1 would be $n$ whereas number of such unordered pairs $(i,j)$ satisfying property 2(i) would be $2c.{{k \choose 2}}$ for every distribution $D$ satisfying the condition in theorem as no 2 servers can 
	can have any common pair.
	
	Similarly $2({{n \choose 2}}- c.{{k \choose 2}})$ unordered pairs exist satisfying property 2(ii) for every distribution $D$ satisfying condition in theorem.
	
	Thus, $\sigma_{D,x}(d)$ is constant for every distribution $D$ with no two servers having more than one job common. 
	
	We can also calculate $\sigma_{D,x}$ as follows:
	
	\begin{align*}
	\sigma_{D,x}(d)= & \frac{1}{{c \choose x}}\left[n.\left(\sum\limits_{t=1}^{r-1} t^2{r \choose (t+1)} {(c-r) \choose (x-t-1)}\right)\\ + 
	&2c.{{r \choose 2}} \left( \sum\limits_{u=2}^{r} \sum\limits_{v=2}^{r}  (u-1)\times (v-1) {{r \choose u}}{{r \choose v}}{{(c-2r) \choose (x-u-v)}}\right)\\+
	& 2\left({{n \choose 2}}- c.{{k \choose 2}}\right) \left(\sum\limits_{u=2}^{r} \sum\limits_{v=2}^{r}  (u-1)\times (v-1) \left({{(r-1} \choose u}{{(r-1) \choose v}}{{(c-2r+1) \choose (x-u-v)}} \\+ &{{(r-1} \choose (u-1)}{{(r-1) \choose (v-1)}}{{(c-2r+1) \choose (x-u-v+1)}}\right)\right)\right] -  \left(\frac{n\sum\limits_{t=1}^{r-1} t{r \choose (t+1)} {(c-r) \choose (x-t-1)}}{{c \choose x}}\right)^2
	\end{align*}

	\end{proof}
}

\remove{
	
\begin{theorem}{\label{min_var}}
	For a given $n$ and $k$, any distribution which satisfies the properties in \ref{setting_condition} with no two servers having more than one common job has always a smaller variance of distinct jobs than any distribution satisfying properties in \ref{setting_condition} having at least one pair of servers with at least two jobs common for every $x$.
\end{theorem}

\begin{proof}
	We know from Equation \eqref{var_sim_exp} of Theorem \ref{var_sim_exp_thm} that: 
	
	\begin{equation*}
	\sigma_{D,x}(d) =  \frac{2.\sum\limits_{p=0}^{r}\mathfrak{n}^D_{p,D}g(p,x) + n\sum\limits_{t=1}^{r-1} t^2{r \choose (t+1)} {(c-r) \choose (x-t-1)}}{{c \choose x}}- \left(\frac{n\sum\limits_{t=1}^{r-1} t{r \choose (t+1)} {(c-r) \choose (x-t-1)}}{{c \choose x}}\right)^2
	\end{equation*}
	
	Now consider $2.\sum\limits_{p=0}^{r}\mathfrak{n}^D_{p,D}g(p,x)$ in the variance expression as it is the only term which depends on the distribution $D$.
	
	
	\remove{	
	Also using the arguments of previous theorem, we show
	
	\begin{equation*}
	\sum\limits_{\substack{\hat{S} \subset \mathcal{S};\\|\hat{S}|=x}} (\mathfrak{n}^D_{i,\hat{S}}-1) \mathbbm{1}_{\mathfrak{n}^D_{i,\hat{S}}>1} (\mathfrak{n}^D_{j,\hat{S}}-1) \mathbbm{1}_{\mathfrak{n}^D_{j,\hat{S}}>1}
	{=}  \sum\limits_{u=2}^{k} \sum\limits_{v=2}^{k}  (u-1)\times (v-1) \sum\limits_{\substack{\hat{S} \subset \mathcal{S};|\hat{S}|=x;\\\mathfrak{n}^D_{i,\hat{S}}=u;\mathfrak{n}^D_{j,\hat{S}}=v}} 1
	\end{equation*}

	Also the value of expression $\sum\limits_{\substack{\hat{S} \subset \mathcal{S};|\hat{S}|=x;\\\mathfrak{n}^D_{i,\hat{S}}=u;\mathfrak{n}^D_{j,\hat{S}}=v}} 1$ for case 1, case 2(i) remain the same as the previous theorem for each pair of $(u,v)$.
	
	However the arguments for case 2(ii) may be different as a pair of jobs $(a_i,a_j)$ may belong to multiple servers.

	Consider pairs of jobs $(i,j)$ such that jobs $a_{i}$ and $a_{j}$ are present together in exactly $m$ servers for every positive integer $m$.
	
	Consider a function $f(i,j) = \sum\limits_{\substack{\hat{S} \subset \mathcal{S};\\|\hat{S}|=x}} (\mathfrak{n}^D_{{i},\hat{S}}-1) \mathbbm{1}_{\mathfrak{n}^D_{{i},\hat{S}}>1} (\mathfrak{n}^D_{{j},\hat{S}}-1) \mathbbm{1}_{\mathfrak{n}^D_{{j},\hat{S}}>1}$
	
	We know that this function remains same irrespective of which $m$ servers have both of them, exactly which servers have only $a_{i}$ or $a_{j}$ due to symmetry over summation.

	Similarly consider a function $f(u,v) = \sum\limits_{\substack{\hat{S} \subset \mathcal{S};\\|\hat{S}|=x}} (\mathfrak{n}^D_{{u},\hat{S}}-1) \mathbbm{1}_{\mathfrak{n}^D_{{u},\hat{S}}>1} (\mathfrak{n}^D_{{v},\hat{S}}-1) \mathbbm{1}_{\mathfrak{n}^D_{{v},\hat{S}}>1}$.

    Let us try to compute the difference between the functions and for this sake let us assume that servers $s_{1},...s_{{m-1}}$ are shared by both the pairs. $s_{{m}}$ is shared by only the pair $(a_{i},a_{j})$, however $s_m$ only contains $a_{u}$, and server $s_{m+1}$ contains only $a_{v}$ but none of the elements in the other pair.
    
    Let $s_{k_1},s_{k_2},...s_{k_{r-m}}$ denote the servers containing only one element from each pair say $a_{i}$ and $a_{u}$, however $s_{l_1},s_{l_2},...s_{l_{r-m}}$ denote the servers containing only one element from each pair say $a_{j}$ and $a_{v}$. This is because job $a_{i}$ occurs without $a_{j}$ in $(r-m)$ such servers and similarly job $a_{u}$ occurs without $a_{v}$ in $(r-m+1)$ such servers which we have satisfied above.
    
    Let us denote the remaining servers by $R$. 
	
	Now let us compare the difference between two functions for each subset of cardinality $x$.
	
	Let us look at 
	
	\begin{equation}{\label{diff_equation}}
	(\mathfrak{n}^D_{{i},\hat{S}}-1) \mathbbm{1}_{\mathfrak{n}^D_{{i},\hat{S}}>1} (\mathfrak{n}^D_{{j},\hat{S}}-1) \mathbbm{1}_{\mathfrak{n}^D_{{j},\hat{S}}>1}- (\mathfrak{n}^D_{{u},\hat{S}}-1) \mathbbm{1}_{\mathfrak{n}^D_{{u},\hat{S}}>1} (\mathfrak{n}^D_{{v},\hat{S}}-1) \mathbbm{1}_{\mathfrak{n}^D_{{v},\hat{S}}>1}
	\end{equation}
	 for all subsets $\hat{S}$ of cardinality $x$.
	
	\begin{enumerate}
		\item If the set $\hat{S}$ does not contain either $s_m$ or $s_m+1$, then we can say that Equation \eqref{diff_equation} goes to zero.
		\item If the set $\hat{S}$ contains both $s_m$ and $s_{m+1}$, then also that Equation \eqref{diff_equation} goes to 0.
	\end{enumerate}
	
Now consider a subset $s$ which does not contain either $s_m$ or $s_{m+1}$ but this set has say $(\alpha+1)$ occurrences of $a_{u}$ and $a_{i}$ and $(\beta+1)$ occurrences of $a_{v}$ and $a_{j}$ for some $\alpha, \beta \geq 0$. 

Now let us look at the set $s \cup \{s_m\}$. For this set Equation \eqref{diff_equation} goes to $(\alpha+1)(\beta+1)- (\alpha+1)\beta$. Similarly, for the set $s \cup \{s_{m+1}\}$, Equation \eqref{diff_equation} goes to $\alpha\beta - \alpha(\beta+1)$.

Thus the sum of Equation \eqref{diff_equation} for these two subsets goes to 1. 

We can also show that if set $s$ contains zero occurrences of $a_{u}$ and $a_{i}$ or zero occurrences of $a_{v}$ and $a_{j}$, the sum of Equation \eqref{diff_equation}  for sets $s \cup \{s_m\}$ and $s \cup \{s_{m+1}\}$ goes to 0, however there are $2{{(c-r-1) \choose (x-1)}}- {{(c-2r+m-1)\choose (x-1)}}$ such subsets. This follows from the fact that such a subset can only exist if it is a subset of  $\{s_{k_1},s_{k_2},...,s_{k_{r-m}}\} \cup R$ or $\{s_{l_1},s_{l_2},...,s_{l_{r-m}}\} \cup R$.  	

Hence, the function $f(i,j)-f(u,v)$ goes to $\left[{{(c-2) \choose (x-1)}}-2{{(c-r-1) \choose (x-1)}}+ {{(c-2r+m-1)\choose (x-1)}}\right]$ which increases with $m$ 
}
 

We know from Theorem \ref{f_diff} that $f(i,j)-f(u,v)$ equals $\left[{{(c-2) \choose (x-1)}}-2{{(c-r-1) \choose (x-1)}}+ {{(c-2r+m-1)\choose (x-1)}}\right]$ which increases with $m$

Now we know $f(i_t,j_t)-f(i_{t-1},j_{t-1})\geq f(i_{t-1},j_{t-1})-f(i_{t-2},j_{t-2})$ implying $f(i_t,j_t)+ f(i_{t-2},j_{t-2}) \geq 2.f(i_{t-1},j_{t-1})$. 

Now for each $t\leq m$, we multiply the inequality by $(m-t+1)$ and sum them for all $t \in \{1,2,...,m\}$ to obtain $mf(i_1,j_1) \leq f(i,j) + (m-1) f(i_0,j_0)$ $\forall m >1$.

\remove{Now consider all possible pairs $K$ each occuring together in $l_K$ servers where $l_K$ is any non-negative integer. Now for any distribution of jobs as per \ref{setting_condition}, $\sum\limits_{K=(i,j); i,j \in \{1,2,...,n\}; i<j} l_K$ is constant. 

For the setting_condition where no two servers have more than one job common, $l_K=0$ or $1$.} 

Now consider any distribution $\tilde{D}$ with at least one pair of servers having more than one job common and let us denote $\hat{D}$ as a distribution with no two servers having more than one job common.

Recall that we denote number of pairs of jobs which occur together in exactly $p$ servers as $\mathfrak{n}^D_{p,D}$ for some distribution $D$.
\remove{
Now let us consider the term  $\sum\limits_{i} \sum\limits_{j} \sum\limits_{\substack{\hat{S} \subset \mathcal{S};\\|\hat{S}|=x}} (\mathfrak{n}^D_{i,\hat{S}}-1) \mathbbm{1}_{\mathfrak{n}^D_{i,\hat{S}}>1} (\mathfrak{n}^D_{j,\hat{S}}-1)$ from $\sigma_{D,x}(d)$.

This term can be written as 2.$\sum\limits_{K=(i,j); i,j \in \{1,2,...,n\}; i<j} f(i,j)+ \sum\limits_i \sum\limits_{\substack{\hat{S} \subset \mathcal{S};\\|\hat{S}|=x}} ((\mathfrak{n}^D_{i,\hat{S}}-1) \mathbbm{1}_{\mathfrak{n}^D_{i,\hat{S}}>1})^2$, the second term can be computed uniquely for every distribution $D$ as in case I of previous theorem. 

The first term can be written as $2.\sum\limits_{p=0}^{r}\mathfrak{n}^D_{p,D}g(p,x)$.}
For every distribution $D$, $\sum\limits_{p=0}^{r}\mathfrak{n}^D_{p,D}= {{n \choose 2}}$ and $\sum\limits_{p=0}^{r}p.\mathfrak{n}^D_{p,D}=c{{k \choose 2}}$ is constant. 

This would imply the number of pairs with exactly one server having both these jobs would be $\sum\limits_{p=1}p.{\mathfrak{n}^D_{p,\tilde{D}}}$ and there would be $({{n \choose 2}} - \sum\limits_{p=1}p.{\mathfrak{n}^D_{p,\tilde{D}}})$ pairs with no server existing which contains both in every distribution $\hat{D}$ since distribution $\hat{D}$ has no pair of jobs which occur together in more than 2 servers.  Also recall that $\tilde{D}$ denotes any distribution which where there exists a pair of jobs which occur together in at least 2 servers.

Thus $\mathfrak{n}^D_{1,\hat{D}}= \sum\limits_{p=1}p.{\mathfrak{n}^D_{p,\tilde{D}}}$ and $\mathfrak{n}^D_{0,\hat{D}}= ({{n \choose 2}} - \sum\limits_{p=1}p.{\mathfrak{n}^D_{p,\tilde{D}}})$ and $\mathfrak{n}^D_{p,\hat{D}}=0$ for every $p >1$ and every distribution $\hat{D}$ with no two servers having more than one job common.

Now using $mf(i_1,j_1) \leq f(i,j) + (m-1) f(i_0,j_0)$, we can show that $\sum\limits_{p=0}^{r}\mathfrak{n}^D_{p,\tilde{D}}g(p,x) \geq \mathfrak{n}^D_{1,\hat{D}} f(i_1,j_1) + \mathfrak{n}^D_{0,\hat{D}} f(i_0,j_0)$ where $\mathfrak{n}^D_{p,\tilde{D}}$ denotes the number of pairs which occur together in $p$ servers in any distribution $\tilde{D}$.


Consider \\
\begin{align*}
\sum\limits_{p=0}^{r}\mathfrak{n}^D_{p,\tilde{D}}g(p,x) \overset{(a)}{=} & \sum\limits_{p=1}\mathfrak{n}^D_{p,\tilde{D}}g(p,x)+  \left({{n \choose 2}}-\sum\limits_{p=1} \mathfrak{n}^D_{p,\tilde{D}}\right)f (i_0,j_0)\\
& \overset{(b)}{\geq} \sum\limits_{p=1}\left(p\mathfrak{n}^D_{p,\tilde{D}}f(i_1,j_1)-(p-1)\mathfrak{n}^D_{p,\tilde{D}}f(i_0,j_0)\right) + \left({{n \choose 2}}-\sum\limits_{p=1} \mathfrak{n}^D_{p,\tilde{D}}\right)f (i_0,j_0)\\
& = \left(\sum\limits_{p=1}p.{\mathfrak{n}^D_{p,\tilde{D}}}\right) f(i_1,j_1)+ \left({{n \choose 2}} - \sum\limits_{p=1}p.{\mathfrak{n}^D_{p,\tilde{D}}}\right) f(i_0,j_0)\\
& = \mathfrak{n}^D_{1,\hat{D}} f(i_1,j_1) + \mathfrak{n}^D_{0,\hat{D}} f(i_0,j_0)
\end{align*}

Note that $(a)$ follows from the fact that $\sum\limits_{p=0}^{r}\mathfrak{n}^D_{p,\tilde{D}} = {{n \choose 2}}$ and $(b)$ follows from $mf(i_1,j_1) \leq f(i,j) + (m-1) f(i_0,j_0)$ for every $m\geq 1$.
 
This would imply that the varince of any distribution $\hat{D}$ would be smaller than that of any distribution $\tilde{D}$.

Thus we prove that any distribution with no more than 2 jobs common between any 2 servers has lower variance than any distribution which has a pair of servers with more than one job common. 
	 
\end{proof}	
}
\remove{
The following corollary follows from the proof of this theorem.
Define $\mathfrak{n}^D_{p,D}$ as in proof of the previous theorem. Recall that $\mathfrak{n}^D_{p,D}$ denotes the number of pairs of jobs which appear together in exactly $p$ servers in a job distribution $D$ as defined in \ref{setting_condition} for every positive integer
$p$.
\begin{corollary}
	Given a distribution satisfies the conditions in \ref{setting_condition}, the variance of the number of distinct jobs received is a function of $\{\mathfrak{n}^D_{p,D}\}_{p\in \mathbb{N}}$ and $x$.
\end{corollary}

\begin{proof}
	This can be proven as we had argued that $g(p,x)$ as defined in previous theorem does not depend on the distribution $D$ as long as jobs $a_{i_p}$ and $a_{j_p}$ are together in exactly $p$ servers.

	We had shown that the numerator in the first term in $\sigma_{D,x}(d)$ can be shown as $2.\sum\limits_{p=0}^{r}\mathfrak{n}^D_{p,D}g(p,x) + n.\left(\sum\limits_{t=1}^{r-1} t^2 {r \choose (t+1)} {(c-r) \choose (x-t-1)}\right)$ implying the variance is just a function of $\{\mathfrak{n}^D_{p,D}\}_{p\in \mathbb{N}}$ and $x$.
	 
\end{proof}

To prove theorem \ref{min_var_gen}, we first show some results using the convexity property of $g(.,x)$ [for a given $x$] as defined in Theorem \ref{result_variance}. Note that $g(m+1,x)-g(m,x)$ as defined in Theorem \ref{result_variance} increases with $p$.



Let us now re-state and prove Theorem~\ref{min_var_gen}.

\begin{theorem*}{}
The least variance of the number of distinct jobs received at the master for any $x \in [2,c-2r+1]$ \remove{ when any subset of servers ($\mathcal{S}$) with cardinality $c-2r+1\leq x>1$ is able to communicate with the master $($denoted by $\sigma_{D,x}(d))$ with equal probability} is attained uniquely by the pairwise balanced job $(n,k,r,c)$ assignment schemes (if it exists) amongst all balanced $(n,k,r,c)$ assignments. 

However, for $x>c-2r+1$, the result on least variance still holds but the uniqueness is not guaranteed.  
\end{theorem*}

Recall the expression of the variance from Theorem \ref{result_variance} and observe that $g(m,x)$ is a convex function in $m$. The essential idea of the proof is to use the convexity property of $g(.,x)$ and show that this expression takes the least value when $\mathfrak{m}^D(m)$ is non-zero for exactly two consecutive values of $m$.

\section{Proof of Theorem \ref{min_var_gen_spec}}{\label{proof_min_var_gen_spec}}

Now we restate and prove Theorem \ref{min_var_gen_spec}.

\begin{theorem*}
The least variance of the number of distinct jobs received at the master for any $x \in [1,c-2r+1]$ \remove{when any subset of servers $\mathcal{S}$ of cardinality $x>1$ is able to communicate with the master (denoted by $\sigma_{D,x}(d)$) with equal probability} is attained uniquely by both the pairwise balanced server $(n,k,k,n)$ assignment schemes and pairwise balanced job $(n,k,k,n)$ assignment schemes (if it exists) amongst all balanced $(n,k,k,n)$ assignments. 
\end{theorem*}

\begin{proof}
    This directly follows from Theorem \ref{min_var_gen} and Claim \ref{job_server_pairs_identical}.
\end{proof}
}

\remove{
\begin{proof}
Consider the definition of $Y^D$ as defined in Claim \ref{var_equal}. Recall that $Y^D$ denotes the number of 
servers where a pair of randomly chosen jobs occur together.

Under $n=c$, theorem \ref{var_equal} would imply that $\text{var}(Y^D)=\sigma_{D,2}(d)$. Also we can show that $\mathbbm{E}[Y^D]= 2k- \mathbbm{E}_{D,2}[d]$ from theorem \ref{mean_equality}.
	
	Thus random variables $2k-d$ and $Y^D$ have the same mean and variance when $x=2$. Recall that $x$ denotes the number of servers we are able to communicate to.
	
	Suppose every pair of servers have $l$ or $l+1$ jobs common. Under this criterion, $d$ (under $x=2$) would be able to take at most 2 consecutive values implying its variance is the least amongst all balanced $(n,k,k,n)$ assignments as it implies $\mathbb{E}_{D,2}[d]=n\left(1-\frac{{{n-k\choose 2}}}{{{n \choose 2}}}\right)$ from Theorem \ref{mean_equality}. Hence, when $d$ only takes two consecutive values, it minimizes the variance amongst all possible distributions of $d$ satisfying the condition on expectation. Since, we know that $\text{var}(Y^D)=\sigma_{D,2}(d)$ and $\mathbbm{E}[Y^D]= 2k- \mathbbm{E}_{D,2}[d]$, we can say that $Y^D$ would have the least variance too. This would imply that $Y^D$ can take exactly at most 2 values, thus every pair of jobs occur together in $l$ or $l+1$ servers.
	
	Thus, using theorem \ref{min_var_gen}, we can say that under this condition, the variance of the number of distinct jobs would be the least for any $x$.   
	
\end{proof}

}